\documentclass{article}
\pdfoutput=1

\usepackage[utf8]{inputenc}

\def\showauthornotes{0}
\def\showkeys{0}
\def\showdraftbox{0}
\def\confversion{0}
\def\widemargin{0}
\def\ipadcompile{1}
\def\todolist{1}

\usepackage{xspace,enumerate}
\usepackage{amsmath,amssymb}
\usepackage{amsthm}
\ifnum\confversion=0
	\usepackage[toc,page]{appendix}
\fi
\usepackage{thmtools}
\usepackage{thm-restate}
\usepackage{color,graphicx}
\usepackage{boxedminipage}
\usepackage{makecell}
\usepackage{tabularx}
\usepackage{csquotes}
\usepackage{scalerel}

\ifnum\widemargin=0
\usepackage[top=1in, bottom=1in, left=1in, right=1in]{geometry}
\else
\usepackage[top=1in, bottom=1in, left=1.25in, right=1.25in]{geometry}
\fi

\ifnum\showkeys=1
\usepackage[color]{showkeys}
\fi

\definecolor{darkred}{rgb}{0.5,0,0}
\definecolor{darkgreen}{rgb}{0,0.35,0}
\definecolor{darkblue}{rgb}{0,0,0.55}

\usepackage[dvipsnames]{xcolor}

\usepackage[pdfstartview=FitH,pdfpagemode=UseNone,colorlinks,linkcolor=NavyBlue,filecolor=blue,citecolor=OliveGreen,urlcolor=NavyBlue,pagebackref]{hyperref}

\usepackage[capitalise,nameinlink]{cleveref}
\usepackage[T1]{fontenc}
\usepackage[utf8]{inputenc}
\usepackage{microtype}

\usepackage{mathtools,dsfont}

\usepackage{palatino}
\usepackage[scaled=.95]{helvet}

\ifnum\ipadcompile=1
\usepackage{eulervm}
\else
\usepackage{eulerpx}
\fi

\usepackage{tcolorbox}

\DeclareMathAlphabet{\mathpazocal}{OMS}{zplm}{m}{n}
\DeclareRobustCommand*{\mathcal}[1]{\mathpazocal{#1}}
\ifnum\showauthornotes=1
\newcommand{\Authornote}[3]{{\sf\small\color{#3}{[#1: #2]}}}
\newcommand{\Authorcomment}[2]{{\sf \small\color{gray}{[#1: #2]}}}
\newcommand{\Authorfnote}[2]{\footnote{\color{red}{#1: #2}}}
\else
\newcommand{\Authornote}[3]{}
\newcommand{\Authorcomment}[2]{}
\newcommand{\Authorfnote}[2]{}
\fi

\ifnum\showdraftbox=1

\else

\fi

\declaretheorem[numberwithin=section]{theorem}
\declaretheorem[sibling=theorem]{lemma}
\declaretheorem[sibling=theorem]{claim}
\declaretheorem[sibling=theorem]{proposition}

\declaretheorem[sibling=theorem]{corollary}

\declaretheorem[sibling=theorem]{summary}

\theoremstyle{definition}
\declaretheorem[sibling=theorem]{definition}
\declaretheorem[sibling=theorem]{remark}

\declaretheorem[sibling=theorem]{observation}

\newtheorem{algo}[theorem]{Algorithm}

\newenvironment{algorithm}[3]
        {\noindent\begin{boxedminipage}{\textwidth}\begin{algo}[#1]\ \par
        {\begin{tabular}{r l}
        \textbf{Input} & #2\\
        \textbf{Output} & #3
        \end{tabular}\par\enskip}}
        {\end{algo}\end{boxedminipage}}

\def\FullBox{\hbox{\vrule width 6pt height 6pt depth 0pt}}

\def\qed{\ifmmode\qquad\FullBox\else{\unskip\nobreak\hfil
\penalty50\hskip1em\null\nobreak\hfil\FullBox
\parfillskip=0pt\finalhyphendemerits=0\endgraf}\fi}

\def\qedsketch{\ifmmode\Box\else{\unskip\nobreak\hfil
\penalty50\hskip1em\null\nobreak\hfil$\Box$
\parfillskip=0pt\finalhyphendemerits=0\endgraf}\fi}

\newcommand{\Caratheodory}{Carath\'eodory\xspace}
\newcommand{\Zemor}{Z\'emor\xspace}

\ifnum\todolist=1
\usepackage{todonotes}
\else
\usepackage[disable]{todonotes}
\fi
\usepackage{scalerel}

\newcommand{\cC}{\mathcal{C}}
\newcommand{\cD}{\mathcal{D}}
\newcommand{\cE}{\mathcal{E}}
\newcommand{\cF}{\mathcal{F}}
\newcommand{\ie}{i.e.,\xspace}

\newcommand{\etal}{et al.\xspace}

\newcommand{\mper}{\,.}

\def\to{\rightarrow}
\def\eps{\varepsilon}
\def\epsilon{\varepsilon}

\def\phi{\varphi}

\def\cal{\mathcal}

\def\implies{\Rightarrow}

\newcommand{\R}{{\mathbb R}}

\newcommand{\F}{{\mathbb F}}

\DeclarePairedDelimiter\parens{\lparen}{\rparen}
\DeclarePairedDelimiter\abs{\lvert}{\rvert}
\DeclarePairedDelimiter\norm{\lVert}{\rVert}

\DeclarePairedDelimiter\braces{\lbrace}{\rbrace}
\DeclarePairedDelimiter\brackets{\lbrack}{\rbrack}
\DeclarePairedDelimiter\angles{\langle}{\rangle}
\DeclarePairedDelimiterXPP\lnorm[1]{}\lVert\rVert{_2}{#1}

\DeclareMathDelimiter{\given}
      {\mathbin}{symbols}{"6A}{largesymbols}{"0C}

\newcommand{\prob}{\mathbb{P}}
\newcommand{\Esymb}{\mathbb{E}}
\newcommand{\tildeE}{\widetilde{\mathbb{E}}}
\newcommand{\Psymb}{\mathbb{P}}

\newcommand{\Varsymb}{\mathrm{Var}}

\DeclarePairedDelimiterXPP{\Prob}[1]
 {\prob}{\lparen}{\rparen}{}
 {\renewcommand{\given}{\;\delimsize\vert\nonscript\;\mathopen{}}#1}

\makeatletter
\def\Pr#1{%
    \ProbabilityRender{\Psymb}{#1}%
}

\def\Ex#1{%
    \ProbabilityRender{\Esymb}{#1}%
}

\def\tildeEx#1{%
    \ProbabilityRender{\widetilde{\Esymb}}{#1}%
}

\def\condPE#1#2{%
	\@ifnextchar\bgroup
	{\ConditionalProbabilityRender{\widetilde{\Esymb}}{#1}{#2}}
	{\ProbabilityRender{\widetilde{\Esymb}}{#1 \given #2}}
}

\def\tildeVar#1{
	\ProbabilityRender{\widetilde{\Varsymb}}{#1}
}

\def\tildeCov#1#2{
	\ProbabilityRender{\tildecov}{#1,#2}
}

\def\ConditionalProbabilityRender#1#2#3#4{
	\renderwithdist{#1}{#2}{#3 \given #4}	
}

\def\ProbabilityRender#1#2{
  \@ifnextchar\bgroup%
  {\renderwithdist{#1}{#2}}
   {\singlervrender{#1}{#2}}
}
\def\singlervrender#1#2{%
   {\mathchoice
       {{#1}\brackets*{#2}}
       {{#1}[ #2 ]}
       {{#1}[ #2 ]}
       {{#1}[ #2 ]}
   }
}
\def\renderwithdist#1#2#3{%
   \@ifnextchar\bgroup
   {\superfancyrender{#1}{#2}{#3}}
   {\mathchoice
      {\underset{#2}{#1}\brackets*{#3}}
      {{#1}_{#2}[ #3 ]}
      {{#1}_{#2}[ #3 ]}
      {{#1}_{#2}[ #3 ]}
     }
}
\def\superfancyrender#1#2#3#4#5{
   \ensuremath{\mathchoice
      {\underset{#1}{{#1}}\left#4 #3 \right#5}
      {{#1}_{#2}#4 #3 #5}
      {{#1}_{#2}#4 #3 #5}
      {{#1}_{#2}#4 #3 #5}
   }
}
\makeatother

 \newcommand\SetSymbol[1][]{%
     \nonscript\:#1\vert
     \allowbreak
     \nonscript\:
     \mathopen{}}
  \DeclarePairedDelimiterX\Set[1]\{\}{%
     \renewcommand\given{\SetSymbol[\delimsize]}
     #1
}

\newcommand{\ip}[2]{\angles*{#1 , #2}}
\newcommand{\inparen}[1]{\parens*{#1}}             
\newcommand{\inbraces}[1]{\braces*{#1}}           
\newcommand{\conv}[1]{\mathrm{conv}\parens*{#1}}

\usepackage{nicefrac}

\let\nfrac=\nicefrac

\newcommand{\one}{{\mathbf{1}}}

\newfont{\inhead}{eufm10 scaled\magstep1}

\newcommand{\calB}{{\cal B}}

\newcommand{\calD}{{\cal D}}

\newcommand{\calH}{{\cal H}}
\newcommand{\calJ}{{\cal J}}
\newcommand{\calL}{{\cal L}}

\newcommand{\calO}{{\cal O}}

\newcommand{\calT}{{\cal T}}
\newcommand{\calS}{{\cal S}}
\newcommand{\calW}{{\cal W}}

\newcommand{\poly}{{\mathrm{poly}}}

\newcommand{\suchthat}{{\;\; : \;\;}}

\newcommand{\sub}{\subseteq}
\newcommand{\defeq}{:=}

\DeclareMathOperator\supp{Supp}
\DeclareMathOperator*{\argmin}{\arg\!\min}
\def\li{\ell}
\def\ri{r}
\def\zee{\mathbf{Z}}
\def\indi#1{{\one \braces*{ #1 } }}

\def\dupPE#1{
	\ProbabilityRender{\widetilde{\Esymb}^*}{#1}
}

\DeclareMathOperator{\tildecov}{\widetilde{\operatorname {Cov}}}
\DeclareMathOperator{\dupCov}{\widetilde{\operatorname {Cov}}_{\widetilde{\Esymb}^*}}

\def\tee{{ \theta }}

\newcommand{\out}{\mathrm{out}}
\newcommand{\inn}{\mathrm{in}}
\newcommand{\dec}{\mathrm{dec}}
\newcommand{\cx}{\cC_{\scaleto{X\mathstrut}{6.5pt}}}
\newcommand{\cz}{\cC_{\scaleto{Z\mathstrut}{6.5pt}}}
\newcommand{\dx}{\cD_{\scaleto{X\mathstrut}{6.5pt}}}
\newcommand{\dz}{\cD_{\scaleto{Z\mathstrut}{6.5pt}}}
\newcommand{\ex}{\cE_{\scaleto{X\mathstrut}{6.5pt}}}
\newcommand{\ez}{\cE_{\scaleto{Z\mathstrut}{6.5pt}}}
\newcommand{\fx}{\cF_{\scaleto{X\mathstrut}{6.5pt}}}
\newcommand{\fz}{\cF_{\scaleto{Z\mathstrut}{6.5pt}}}
\newcommand{\hx}{H_{\scaleto{X\mathstrut}{6.5pt}}}
\newcommand{\hz}{H_{\scaleto{Z\mathstrut}{6.5pt}}}
\newcommand{\subs}[2]{{#1}_{\scaleto{{#2}\mathstrut}{6.5pt}}}
\newcommand{\dxh}{\overline{h}}
\newcommand{\dxz}{\overline{z}}

\newcommand{\distL}[1]{\Delta_L\parens*{#1\,,h}}
\newcommand{\distLperp}[1]{\Delta_{L, \cz^\perp}\parens*{#1\,,h}}
\newcommand{\distR}[1]{\Delta_{R}\parens*{#1\,,h}}

\newcommand{\dist}[1]{\Delta \parens*{ #1 }}
\newcommand{\partmin}{\psi}
\newcommand{\disR}[2]{\Delta_{R}\parens*{#1\,,#2}}

\newcommand{\phix}{\subs{\varphi}{X}}
\newcommand{\phiz}{\subs{\varphi}{Z}}
\newcommand{\wphix}{\widetilde{\phix}}
\newcommand{\wphiz}{\widetilde{\phiz}}
\newcommand{\unconx}{\subs{\varphi}{X}^{\raisebox{0.3ex}{$\scriptscriptstyle-1$}}}
\newcommand{\unconz}{\subs{\varphi}{Z}^{\raisebox{0.3ex}{$\scriptscriptstyle-1$}}}
\newcommand{\embed}{\chi}
\newcommand{\chiTan}{\chi}
\DeclareMathOperator{\fold}{\operatorname{Fold}}

\title{List Decodable Quantum LDPC Codes}
\author{
Thiago Bergamaschi\thanks{{\tt UC Berkeley}. {\tt thiagob@berkeley.edu}. }
\and
Fernando Granha Jeronimo\thanks{{\tt University of Illinois Urbana-Champaign}. {\tt granha@illinois.edu}. }
\and
Tushant Mittal\thanks{{\tt Stanford University}. {\tt tushant@uchicago.edu}. Partly supported by NSF grant CCF-2326685.}
\and
Shashank Srivastava\thanks{{\tt DIMACS \& IAS}. {\tt shashanks@ias.edu}. Partly supported by the NSF grant CCF-2326685.}
\and
Madhur Tulsiani\thanks{{\tt Toyota Technological Institute at Chicago}. {\tt madhurt@ttic.edu}. Supported by the NSF grant CCF-2326685.} 
}
\date{}

\begin{document}

\maketitle
\thispagestyle{empty}

\medskip
We give a construction of Quantum Low-Density Parity Check (QLDPC) codes with near-optimal rate-distance tradeoff and efficient list decoding up to the Johnson bound in polynomial time. 
Previous constructions of list decodable good distance quantum codes either required access to a classical side channel or were based on algebraic constructions that preclude the LDPC property.


Our construction relies on new algorithmic results for codes obtained via the quantum analog of the distance amplification scheme of Alon, Edmonds, and Luby [FOCS 1995].  These results are based on convex relaxations obtained using the Sum-of-Squares hierarchy, which reduce the problem of list decoding the distance amplified codes to unique decoding the starting base codes. Choosing these base codes to be the recent breakthrough constructions of good QLDPC codes with efficient unique decoders, we get efficiently list decodable QLDPC codes.




\newpage

\pagenumbering{roman}
\tableofcontents
\clearpage

\pagenumbering{arabic}
\setcounter{page}{1}

\newpage


\section{Introduction}
The area of quantum error correction has had tremendous progress in recent years, particularly in the construction of quantum low-density parity-check (QLDPC) codes. These are codes where membership can be tested via checks acting only on a small number of qubits, which is often an important property for physical implementations. 
Starting with the classic toric code by Kitaev~\cite{Kitaev03}, a sequence of works~\cite{TZ14, EKZ20, KT21, HHOD21, PK22, LZ22} has led to the construction of \emph{asymptotically good} QLDPC codes, with constant rate and constant relative distance. 

Specifically, the above constructions yield a special form of quantum code known as
Calderbank--Shor--Steane (CSS) code, which can be specified by a pair of subspaces $\cx,\cz
\subseteq \F_q^n$ satisfying $\cz^{\perp} \subseteq \cx$ (which also implies $\cx^{\perp} \subseteq
\cz$). In this work, we will always take quantum codes to mean quantum CSS codes and refer to them
as QLDPC codes when $\cx^{\perp}$ and $\cz^{\perp}$ have generating sets consisting of sparse vectors. The code $\cC =
(\cx, \cz)$ is said to have blocklength $n$ and alphabet size $q$, with the relative distance $\delta$ and the rate $\rho$ defined as,
\[
\delta  ~=~ \frac{1}{n} \cdot \min\inbraces{\abs{c} ~\Big\vert~ c \in (\cx \setminus \cz^{\perp}) \cup (\cz \setminus \cx^{\perp})} 
\qquad \text{and} \qquad
\rho ~=~ \frac{1}{n} \cdot \inparen{\dim (\cx) - \dim(\cz^\perp)} \mper
\]
The codes are said to be \textit{good codes} if both $\delta$ and $\rho$ are constants independent of $n$.
 
The first construction of good QLDPC codes was obtained by the recent breakthrough construction of
Panteleev and Kalachev~\cite{PK22}, also closely related to the independent construction
of good locally testable codes by Dinur, Evra, Livne, Lubotzky, and Mozes~\cite{DELLM22}. Subsequent variants of their
constructions by Leverrier and \Zemor~\cite{LZ22, LZ23:decodable} and by Dinur, Hsieh, Lin, and Vidick~\cite{DHLV23} have also obtained simpler descriptions and stronger properties, including \emph{algorithmic} guarantees such as the existence of linear-time (and even parallelizable) unique-decoding algorithms, for decoding from errors up to a constant fraction of the blocklength. 
\paragraph{Improved error correction via list decoding.} 
Given the existence of codes with constant relative distance and associated unique-decoding algorithms, it is natural to {consider the} \emph{maximum} fraction of errors that can be efficiently corrected and whether even stronger forms of decoding are possible for such codes. 

For a code with relative distance $\delta$, unique-decoding can recover from, at most, a $\delta/2$ fraction of errors (although existing unique-decoding algorithms for quantum codes do not necessarily achieve this threshold). 
For several classical codes, it is possible to correct significantly more than $\delta/2$ fraction
of errors by relaxing the decoding task to \emph{list decoding}~\cite{G01, Guruswami:survey}, where the goal is not to output a single codeword but possibly a (small) list of codewords within a given error radius $\tau$ of the received word.  
In addition to allowing for recovery from $\tau > \delta/2$ fraction of errors, this is also helpful for several applications in complexity theory~\cite{Trevisan05} where tolerating a larger error radius is significantly more important than recovering a unique codeword.

In the quantum case, no-cloning restrictions forbid duplication of quantum states and require the task of list decoding to be defined more carefully. The appropriate analog in the quantum case is actually the \emph{classical} task of recovering a list of error patterns with at most {$\tau$-fraction of} errors that can lead to a given received word (or rather, a given syndrome, which corresponds to the output of the parity checks). 
The question of constructing list decodable quantum codes (but not necessarily QLDPC codes) was considered by Leung and Smith~\cite{LS08}, and more recently, quantum codes admitting efficient list decoding algorithms were also constructed in the work of Bergamaschi, Golowich, and Gunn~\cite{BGG22}.

Another form of decoding for quantum codes, which can go beyond the unique-decoding radius $\delta/2$, is where one is still trying to recover a single quantum state but is allowed to make an exponentially small (in the block-length $n$ of the code) error in the output quantum state. This notion of approximate quantum error-correcting codes (AQECCs) was previously considered by Cr\'{e}pau, Gottesman, and Smith~\cite{CGS05}, and has found renewed recent interest \cite{HP20, Ber23, YYGL24}. 
The work of Bergamaschi, Golowich, and Gunn was also motivated by the construction of AQECCs and shows that existing constructions of purity testing codes~\cite{BCGST02} and robust secret sharing schemes~\cite{CDDFS15} can be combined with list decodable quantum codes, to obtain AQECCs (with the same decoding radius).
Thus, constructions of quantum codes with efficient list decoding can be used for both the above forms of error correction beyond the unique decoding radius.

In light of the above discussion, it is desirable to have quantum codes that are both, efficiently list decodable and LDPC. Constructing such list decodable LDPC codes is challenging, even for classical codes. This is because most 
list decoding techniques rely upon algebraic structure, which does not necessarily allow for the LDPC property.

\subsection{Our Results} 
We give the first construction of QLDPC codes that are efficiently list decodable beyond the unique decoding radius. We can decode upto the \textit{Johnson bound} defined as $\mathcal{J}(\delta) := 1 - \sqrt{1-\delta}$. It is not known how to construct QLDPC codes, or indeed even classical LDPC codes, which can be efficiently list decoded beyond the Johnson bound.

\begin{theorem}[Informal version of \cref{thm:near_mds_main}]\label{thm:intro_main}
	For any constant $0 <\rho < 1$, and small enough $\eps_1,\eps_2 > 0$, there is an infinite family of quantum LDPC codes over a constant-sized alphabet, $q(\eps_1)$, such that:
	\begin{enumerate}[(i)]
        \item Each code can be constructed deterministically in time polynomial in its blocklength, $n$.  
        \item The rate of each code is at least $\rho$, and the relative distance is $\delta \geq \frac{1-\rho-\eps_1}{2}$.
		\item Each code can be list decoded from radius $\calJ\parens[\big]{\frac{1-\rho - \eps_1}{2}}-\eps_2$ in time, $n^{\calO_{\eps_1}(1/\eps_2^4)}$.
	\end{enumerate}
\end{theorem}

Our construction relies on a quantum analog of the distance amplification procedure due to Alon, Edmonds, and
Luby~\cite{AEL95}. The  Alon--Edmonds--Luby (AEL) procedure takes a balanced bipartite graph ---
$G=(L,R,E)$ with $\abs{L}=\abs{R}=n$ and degree $d$ --- and two quantum CSS codes: an ``inner code''
$\cC$ of blocklength $d$ and an ``outer code'' $\cD$ of blocklength $n$. 
It combines these via concatenation and folding to construct a new quantum code $\cF$ with improved distance properties. Additionally, if the outer code $\cD$ is QLDPC and the degree, $d$, is a constant, then the final code $\cF$ is also QLDPC. 

 We show that this procedure yields a general method for obtaining list decodable codes, using only \emph{unique-decodability} of the starting outer code. 
 
%
%


%
\begin{theorem}[Informal version of \cref{thm:list_decoding_ael}]
 Let $\cF$ be the CSS code obtained by applying quantum AEL amplification to the outer code $\cD$ and inner code $\cC$ using an expander graph such that $\cF$ has distance $\delta \approx \delta (\cC)$. If the outer code has an efficient unique-decoder, then $\cF$ can be efficiently list decoded up to error radius $\tau \approx \mathcal{J}(\delta) > \delta/2$.
\end{theorem}
%
%
%
%
 We obtain \cref{thm:intro_main} by picking the outer codes to be the recent constructions of unique-decodable QLDPC codes~\cite{LZ23:decodable}. We choose the (constant-sized) inner code, $\cC$, to be an optimal CSS code matching the Singleton bound, and get $\delta \approx \delta(\cC) \approx 1/2$ and $\tau \approx \mathcal{J}(\delta) \approx 0.293 > 0.25$.

\subsection{Technical Overview}

Our results are based on the framework developed by~\cite{JST23} for list decoding of (classical) codes up to
the Johnson bound, using \enquote{covering lemmas} and \enquote{proofs of distance} that can be captured by convex relaxations in the Sum-of-Squares (SoS) hierarchy of semidefinite programs. 
The covering lemma ensures that the solution to the SoS convex relaxation captures sufficient
information about each element of the list. Given the covering, the list elements are then isolated
using the fact that different codewords are sufficiently far from each other.
The SoS framework requires expressing the proof of distance for the code being decoded, in terms of inequalities obtained via non-negativity of sum-of-squares of low-degree polynomials (in formal variables corresponding to codeword symbols). 

For the case of codes obtained via the AEL distance amplification procedure, their framework gives a reduction --- using convex relaxations in the SoS hierarchy --- from the task of list decoding the resulting codes to that of unique-decoding the outer code.
The reduction relies on the observation that the distance proofs for codes constructed via AEL can
be expressed as the sum of squares of low-degree polynomials. 
This is because AEL distance amplification is carried out using constant-degree expander graphs, and distance proofs are based on spectral inequalities (which correspond to quadratic forms of positive semidefinite matrices arising from the underlying graph).
Our result can be viewed as a quantum analog of this reduction.

\paragraph{{Quantum Extensions and challenges}} Extending the above framework for quantum codes faces some important bottlenecks. There are two key aspects to this generalization: one is that the notion of distance itself is different for quantum codes, and the second is that we need a proof that is expressible as sum-of-squares of low-degree polynomials.

\begin{description}
	\item[Quantum Distance] For a classical code $C \subseteq \F_q^n$, the distance between $u,v \in C$ is defined as the Hamming distance $\Delta(u,v)$ between these vectors, which is easily expressible as a low-degree polynomial in (a real-valued embedding of) the coordinates of $u$ and $v$. 
For a CSS quantum code $\cF = (\fx,\fz)$ the distance between (say) $u,v \in \fx$ is defined as the distance between \emph{cosets} of $\fz^{\perp}$ containing $u$ and $v$ i.e. $\min_{w \in {\cF^\perp_Z}} \Delta(u, v-w)$. 

	\item[Low-Degree Proof] A common technique in proving distance between $u,v \in \fx$ that belong to different cosets of $\fz^\perp$ is to replace $u$ by $u^*$, where $u^*$ is the closest element of $u+\fz^\perp$ to $v$. This optimality of $u^*$ with respect to distance from $v$ gives it extra structural properties, which are then exploited to prove distance $\Delta(u^*,v)=0$ or $\Delta(u^*,v) > \delta $. When working with SoS, however, the trivial statement \enquote{$\Delta(u,v) \geq \Delta(u^*,v)$} need not have a low-degree proof!

\end{description}

The change in distance notions actually presents a significant barrier to implementing the ``proofs to algorithms" approach. In fact, an important SoS lower bound by Hopkins and Lin~\cite{HL22} shows that even the \emph{definition} of quantum distance cannot be captured by SoS, and even exponential-sized semidefinite programs from the SoS hierarchy cannot distinguish codewords at distance zero (in the same coset) from those that are far apart! This is overcome in the present work by a \emph{new proof} of distance amplification for quantum codes via the AEL procedure, which lets us separate ``local" vs "global" aspects of the distance for the resulting codes.
%

Our main technical contribution is thus to give a low-degree proof of the distance bound for the quantum version of the AEL procedure. The other known proof of quantum AEL~\cite{BGG22} uses the unique decoding radius and avoids dealing with $\fz^{\perp}$. While it is an easy argument that the unique-decoding radius is half of the distance, this fact is not necessarily captured as a statement in terms of low-degree polynomials.

In contrast, our proof uses the observation that the space $\fz^{\perp}$, for codes obtained by AEL amplification, has additional structure and can be decomposed as $\fz^{\perp} = \cF_{\out} + \cF_{\inn}$. The distance minimization is then done in two steps corresponding to the two spaces. 


We define the notion of a \textit{partial minimizer} by utilizing decomposition of $\cF_{\inn}$ into local subspaces, $\cF_{\inn} = \F_q^n \otimes \cz^{\perp}$. For any fixed pair $u,v \in \fx$, the \emph{partial minimizer} maps $u$ to $u' \in u + \cF_{\inn}$, instead of the optimal $u^*$ which is closest to $v$. This $u'$ is not necessarily the closest to $v$, but, crucially for us, enables a low-degree proof for the statement, \enquote{$\Delta(u,v) \geq \Delta(u',v)$.} 



The above SoS solution can then be "rounded" (via standard techniques) to obtain a vector within the \emph{unique} decoding radius. We augment the SoS algorithm with the unique decoding algorithm for the outer code, to handle the minimization over $\cF_\out$, which presented a barrier for the SoS algorithm itself.

\subsection{Related work} 
Our work is closely related to that of Bergamaschi, Golowich, and Gunn~\cite{BGG22} which considered AEL amplification in the context of quantum codes. They give multiple code constructions that are either (capacity-achieving) list decodable quantum codes or QLDPC codes (with optimal rate-distance tradeoffs), but not both simultaneously.

Their results on distance amplification yields QLDPC codes nearing the Singleton bound, that are linear-time unique decodable. However, these codes can only be list decoded with access to a classical side channel, where one can transfer classical bits without errors. To obtain list decodable codes, they give another construction by applying the AEL procedure using (CSS codes obtained from) folded Reed--Solomon codes~\cite{GuruswamiR08} as outer codes, which are known to have optimal list decoding properties as a function of the rate. However, since folded Reed--Solomon codes are not LDPC, the list decodable quantum codes thus obtained are not QLDPC codes.
%

%

\paragraph{Proofs to Algorithms paradigm}Spectral algorithms and semidefinite programming for list decoding have been used since the works of Guruswami and Indyk~\cite{GI03} and Dinur \etal~\cite{DHKLNTS19}. More recently, the ``proofs to algorithms'' paradigm has been used in the application of the sum-of-squares method to several statistical and combinatorial problems~\cite{FKP19}. In the context of list decoding, this paradigm was used by Richelson and Roy~\cite{RR23} for decoding the $\epsilon$-balanced code construction of Ta-Shma~\cite{TS17} via an SoS implementation of the distance proof. A similar framework was used by~\cite{JST23} to list decode several classical codes based on spectral expanders, including the AEL distance amplification procedure.

\paragraph{Lower bounds}
Another important line of work is on understanding the limitations of SoS-based algorithms. The inability of SoS to reason about vector spaces over finite fields is a property used to
construct lower bounds for the SoS hierarchy, ruling out even subexponential time algorithms
within this hierarchy, for several problems~\cite{Grigoriev01, Schoenebeck08, Tulsiani09, DFHT21, HL22}. In particular, the recent result of Hopkins and Lin~\cite{HL22} gives an explicit lower bound instance using an expander-based QLDPC code. While their work shows that SoS relaxations cannot distinguish between $\fx$ and $\fz^{\perp}$ for some quantum code $\cF = (\fx, \fz)$, our work can be viewed as proving that this is indeed possible with access to a decoder for the outer code, when $\cF$ is obtained via the AEL construction.

%
%
%
\paragraph{Discussion and Open Problems.}
Our work gives a construction of quantum LDPC codes decodable up to the Johnson bound, and it is an intriguing open problem to construct LDPC codes decodable beyond this threshold, even in the case of classical codes.
%
One concrete direction arising from the present work is to find a \emph{better reduction} between the decodability of codes constructed by amplification procedures, such as AEL, and the starting points of such constructions. 
The current reduction using the SoS hierarchy relies on the expansion used in the amplification procedure and yields decodability up to the Johnson bound. However, it is conceivable that a reduction using additional structure from (a variation of) the amplification procedure can go beyond this threshold.

We also note that while the running times of the algorithms obtained in this work are polynomial
(when $q,d$ are taken to be constants), they are very far from the highly efficient linear-time
unique decoding algorithms that are known for classical and quantum LDPC codes. However, we hope
that understanding the structures that facilitate list decoding via SoS can point the way towards
fast spectral algorithms relying on these structures, as has been the case in the past  for statistical problems such as tensor decomposition~\cite{HSSS16},  and for decoding some classical codes based on expanders~\cite{JST21}. 

\paragraph{Outline}  We start by defining the task of list decoding for vector space CSS codes. We also define the notion of duality preserving maps, which we then use in~\cref{sec:ael} to formalize the central code operations--- \textit{code concatenation} and AEL \textit{amplification}--- in an explicit linear algebraic way. \cref{sec:distance,sec:sos_distance} form the key contribution of our work in which we give a new proof of quantum AEL distance amplification. The main idea is to introduce the \textit{partial minimizer} using which we give a distance proof (\cref{sec:distance}). After setting up some additional preliminaries in \cref{sec:sos_prelims}, we show how our distance proof easily generalizes to the SoS setting~\cref{sec:sos_distance}).
Finally, in~\cref{sec:decoding}, we give a polynomial-time list decoding algorithm by plugging our SoS proof of distance into the machinery of~\cite{JST23}.

\section{Preliminaries}
\label{sec:prelims}

Throughout the paper, we will work with $\F_q$-vector spaces with the standard basis (unless specified otherwise). The spaces will be equipped with the Hamming metric and the canonical bilinear form $\angles{u,v} = \sum_i u_iv_i$, with respect to this basis. We define, $V^\perp = \Set{w\mid \ip{v}{w} = 0,\, \forall\, v \in V}$. We will also work with $\R$-vector spaces equipped with the expectation inner product and norm.

\subsection{Quantum CSS codes and list decoding}
%
%
Calderbank--Shor, and Steane (independently) showed that a pair of subspaces over $\F_q$ define a
quantum code if they satisfy a certain orthogonality condition. This construction is known as the CSS construction, and CSS codes form a subclass of \textit{stabilizer codes}.    
 



\begin{definition}[CSS Codes] Let $\cx,\cz \subseteq \F_q^n$ be $\F_q$-linear subspaces such that $\cz^\perp \subseteq \cx$. Then, $\cC = (\cx, \cz)$ defines a $[[n,k,\delta n]]_{q}$ quantum code, where 
\[ 
k =  \dim (\cx) - \dim(\cz^\perp), \;\; \delta = \frac{1}{n} \min \braces[\big]{ \abs*{v}\;\big\vert\; v \in (\cx\setminus \cz^\perp) \cup (\cz\setminus \cx^\perp)}.
\]
 The CSS code is \textit{low-density parity check} (LDPC) if there exist (row and column)-sparse parity check matrices, $\hx, \hz$, over $\F_q$ such that $\cx = \ker \hx$ and $\cz = \ker \hz$.
\end{definition}

 

\paragraph{Vector space CSS codes}
The works~\cite{BGG22, GG23} show that the CSS construction can be generalized to the setup where the $\F_q$-linear subspaces have coefficients in a vector space $\F_q^b$.  This generalization enables the extension of the classical coding theoretic operation of folding to the quantum setup. We will work with this general definition throughout the paper. 

\paragraph{Folding} \textit{Folding} is an element-wise syntactic operation in which vectors in $\F_q^{bn}$ are viewed as vectors in $\parens*{\F_q^b}^n$. Formally, let $V \subseteq \F_q^{bn}$ be a vector space over $\F_q$. For a vector $v \in V$, denote \[\fold(v) = (v^{(1)}, \cdots, v^{(n)}) \in \parens*{\F_q^b}^n \text{ where } v^{(i)} = (v_{(i-1)b+1},\cdots, v_{ib})\in \F_q^b.\] 


\begin{definition}[Vector Space CSS code] Let $d$ be a positive integer and let $\cC = (\cx,\cz)$ be a $[[nb,k,\delta nb]]_q$ CSS code. Then, $\fold(\cC)$ defines a $[[n,k/b,\delta' n]]_{q,b}$ vector space CSS code wherein, 
 \[ 
\delta' = \frac{1}{n} \min \braces[\big]{ \abs*{\fold(v)}\;\big\lvert\; v \in (\cx\setminus \cz^\perp) \cup (\cz\setminus \cx^\perp)}.
\]
The weight, $\abs*{c}$, now is the Hamming weight over the alphabet $q^b$, \ie $|\fold(v)| = \abs{\braces{i \mid v^{(i)} \neq 0}}$. 	
\end{definition}

We will use the notation $[[n,k/b,\delta' n]]_{q,b}$ to specify that the code consists of $\F_q$-linear spaces that have been folded in blocks of size $b$. Unless specified otherwise, the folding will be done sequentially to the coordinates according to the fixed basis. {We will use the shortened notation $[[n,k]]$ when we do not require to address the distance of the code.  We will drop the notation $\fold(\cC)$ when we explicitly mention that a code is a vector space CSS code.}

\begin{remark} 	
 The \enquote{folded dimension} changes to $k/b$ to make it consistent with the classical notion. However,  both dimension and blocklength change by the same factor $b$, so that the rate is unchanged. Moreover, the dimensions as $\F_q$-subspaces remain unchanged after folding. 
\end{remark}

\paragraph{List decoding} We now formalize the notion of list decoding for quantum CSS codes and the folded codes. This is inspired by the classical definition but there are a couple of crucial changes in the quantum setup, (i) the input is no longer a corrupted codeword but a \textit{syndrome}, and (ii) the output is required to be a list \textit{containing} the list of possible errors. The first constraint is necessitated by the no-cloning theorem whereas the second relaxation is required as pruning the list can be hard, unlike in the classical case.

\begin{definition}[List of codewords]
Let $\cC = (\cx,\cz)$ be a vector space CSS code over $\parens*{\F_q^d}^n$. Let $B(g,\tau)$ denote a ball of fractional radius $\tau$, in the Hamming metric over $q^d$, around a vector $g$. For any pair of vectors $\subs{g}{X}, \subs{g}{Z} \in \parens*{\F_q^d}^n$, we define the following lists as lists of cosets of codewords, 
\begin{align*}
\calL_X (\subs{g}{X}, \tau) ~&=~ \braces[\Big]{\subs{h}{X} + \cz^\perp\, \mid \,\subs{h}{X} \in \cx,\,\; B(\subs{g}{X},\tau) \cap \subs{h}{X} + \cz^\perp \neq \emptyset},\\
\calL_Z (\subs{g}{Z}, \tau) ~&=~ \braces[\Big]{\subs{h}{Z} + \cx^\perp\, \mid \,\subs{h}{Z} \in \cx,\,\; B(\subs{g}{Z},\tau) \cap \subs{h}{Z} + \cx^\perp \neq \emptyset},\\
\calL (\subs{g}{X}, \subs{g}{Z},\tau) ~&=~ \calL_X (\subs{g}{X},\tau)\, \times\, \calL_Z (\subs{g}{Z},\tau).
\end{align*}
\end{definition}


We also define $\calL_e (\subs{g}{X}, \subs{g}{Z},\tau)$ as a list of cosets of errors, which is just the codeword list $\calL(\subs{g}{X}, \subs{g}{Z},\tau)$ shifted by $\subs{g}{X}$ and $\subs{g}{Z}$ respectively. When decoding from syndromes, all information about the original codeword is lost since the syndrome only depends on the error pattern. Therefore, one can only hope to output a list of errors rather than a list of codewords. In other words, since syndrome is invariant to translation by codewords, so should the output of a list decoding algorithm, and $\calL_e (\subs{g}{X}, \subs{g}{Z},\tau)$ is the translation-invariant version of $\calL (\subs{g}{X}, \subs{g}{Z},\tau)$.
\begin{align*}
\calL_e (\subs{g}{X}, \subs{g}{Z},\tau) ~&=~ \braces[\Big]{\parens[\big]{\subs{h}{X} -\subs{g}{X}  + \cz^\perp, \subs{h}{Z} - \subs{g}{Z}+ \cx^\perp}\, \big\vert \, \parens[\big]{\subs{h}{X} + \cz^\perp, \subs{h}{Z} +\cx^\perp} \in \calL (\subs{g}{X}, \subs{g}{Z},\tau)}\\[3pt]
~&=~ \parens*{\calL_X (\subs{g}{X},\tau) -\subs{g}{X}} \times \parens*{\calL_Z (\subs{g}{Z},\tau)-\subs{g}{Z}}\\[4pt]
~&=~\calL_e (\subs{g}{X}', \subs{g}{Z}',\tau) \qquad \text{for any}\qquad  \subs{g}{X}- \subs{g}{X}' \in \cx, \, \subs{g}{Z}- \subs{g}{Z}' \in \cz.
\end{align*} 
The last equality follows as the codes $\cx, \cz$ are linear and the Hamming metric is translation-invariant.


\begin{definition}[List Decodable vector space CSS codes, also in~\cite{BGG22}]
A quantum CSS code $\cC$ is $(\tau,L)$-\textit{list decodable} if for every $(\subs{g}{X}, \subs{g}{Z})$, the list size is bounded, \ie $\abs{\calL (\subs{g}{X}, \subs{g}{Z},\tau)} \leq L$.

 Fix a pair of parity check matrices $(\hx,\hz)$. We say that a code is \textit{efficiently list decodable} upto fractional radius $\tau$ if given $(\hx\subs{g}{X}, \hz\subs{g}{Z})$ such that $\subs{g}{X},\subs{g}{Z} \in B(0,\tau)$, there exists a $\poly(n)$-time algorithm that outputs a list that contains $\calL_e(\subs{g}{X}, \subs{g}{Z},\tau)$.  
\end{definition}

 \begin{observation} Assume that for any $(\subs{g}{X}, \subs{g}{Z})$, one can output the lists $\calL_X (\subs{g}{X}, \tau)$ and $\calL_Z (\subs{g}{Z}, \tau)$ in $\poly(n)$-time.  Then, the quantum CSS code is efficiently list decodable. 
 \end{observation}
 \begin{proof}
Given, $(\hx\subs{g}{X}, \hz\subs{g}{Z})$ one use Gaussian elimination to compute $(\subs{g}{X}', \subs{g}{Z}')$ such that $\subs{g}{X} - \subs{g}{X}' \in \cx$ and $\subs{g}{Z} - \subs{g}{Z}' \in \cz$. By the assumption we can compute $\calL_X (\subs{g}{X}',\tau)$ and $\calL_Z (\subs{g}{Z}',\tau)$, and therefore output, 
$\calL_e (\subs{g}{X}, \subs{g}{Z},\tau) = \parens[Bigg]{\calL_X (\subs{g}{X}',\tau) -\subs{g}{X}'} \times \parens[Bigg]{\calL_Z (\subs{g}{Z}',\tau)-\subs{g}{Z}'}$.
 \end{proof} 
 
 In summary, we can reduce the task of list decoding the classical codes $\fold(\ex)$ and $\fold(\ez)$ but \textit{upto cosets}. Moreover, our decoder will be symmetric and henceforth we will focus on the task of $X$-decoding.

 \begin{summary}[List decoding vector space codes]\label{sum:list_dec_vector_space}
 Let $\fold(\cE)$ be a vector space CSS code over $\parens*{\F_q^d}^n$ with given parity-check matrices $(\hx,\hz)$. The task of efficient list decoding $\fold(\cE)$ upto radius $\tau$ reduces to the following two tasks,
 \begin{itemize}
 	\item \textup{\textbf{X-decoding}}: Given as input $\subs{g}{X} \in \parens*{\F_q^d}^n$, output a list, $\calL_X'$ of cosets of codewords such that, 
 	\[ \calL_X (\subs{g}{X}, \tau)~\subseteq~ \calL_X' ~\subseteq~ \fold(\ex)/\fold(\ez^\perp).\]   
 	\item \textup{\textbf{Z-decoding}}: Given as input $\subs{g}{Z} \in \parens*{\F_q^d}^n$, output a list, $\calL_Z'$ of cosets of codewords such that,
 	\[ \calL_Z (\subs{g}{Z}, \tau)~\subseteq~ \calL_Z' ~\subseteq~ \fold(\ez)/\fold(\ex^\perp).\] 
 \end{itemize}
 	
 \end{summary}

\subsection{Duality preserving maps} 

To generalize the notion of concatenation to quantum CSS codes, we will need the notion of duality-preserving maps. These are needed to properly define the concatenated code such that the orthogonality constraint, $\cz^\perp\subseteq \cx$, is satisfied, thereby defining a quantum CSS code. This has been used in earlier works, for example, see~\cite{Ham08}. 

{A bilinear map $\ip{\cdot}{\cdot} :V \times W \to \F_q $ over $\F_q$-vector spaces $V,W$, is \textit{non-degenerate} if for any non-zero $x \in V$, the map $\ip{x}{\cdot} : W\to \F_q$ is not identically zero. Similarly,  for a non-zero $y \in W$, the map $\ip{\cdot}{y} : V\to \F_q$ is not the zero map.
}
\begin{definition}[Dual Systems and Basis]\label{def:dual}
	Let $V,W$ be $\F_q$-vector spaces of equal dimension. Let $\ip{\cdot}{\cdot}$ be a non-degenerate bilinear map $V \times W \to \F_q$. Then, $(V,W,\ip{\cdot}{\cdot})$ defines a \textit{dual system}. A basis $\{v_1,\cdots, v_b\}$ of $V$, and $\{w_1,\cdots, w_b\}$ of $W$ is said to be \textit{dual} if $\ip{v_i}{w_j} = 
	\delta_{ij}$ for all $i,j \in [b]$.\end{definition}

{The spaces $V,W$ are called \textit{dual spaces} as the bilinear map gives an isomorphism $W \to V^*$ defined as $w \mapsto \ip{w}{\cdot}$, which also proves the existence of such basis. }This map is injective as the bilinear map is non-degenerate. Now, one can use the canonical dual basis of $V^*$. Using a dual basis, one can construct \textit{duality-preserving maps} which are what we will need. 
%
\begin{definition}[Duality preserving map]\label{def:dual_maps}
	Let $(V_1,V_2, \ip{\cdot}{\cdot}_V)$ and $(W_1,W_2, \ip{\cdot}{\cdot}_W)$ be two dual systems. A pair of linear maps $(\phi_1,\phi_2)$ where $\phi_i: V_i \to W_i $ are \textit{duality preserving} if,
	\begin{equation}
   \ip{u}{v}_V = \ip{\phi_1(u)}{\phi_2(v)}_{W}  \; \; \forall\, u,v \in V.
   \end{equation}
\end{definition}

\begin{claim}\label{claim:dual_map}
	Let $(V_1, V_2)$ and be $(W_1, W_2)$ dual systems along with their pairs of dual bases. Then, $\phi_i: V_i \mapsto W_i$ that acts as identity with respect to these dual bases is a duality-preserving map.  
\end{claim}
\begin{proof}
Since the condition is bilinear, it suffices to prove it for any pair of basis vectors of $V_1, V_2$. 
\begin{equation*}
   \ip{u_i}{v_j}_V = \delta_{ij} =    \ip{x_i}{y_j}_V = \ip{\phi_1(u_i)}{\phi_2(y_j)}_{W}.\qedhere
   \end{equation*} 
\end{proof}


{ We now define two different dual systems we work with. Both use the canonical bilinear form over $\F_q^n$, albeit for different $n$. The first is the space $(\F_q^b,\F_q^b)$ which forms a dual system with the elementary basis as a dual basis. The second system will be subspaces of a CSS code. We now prove that this second system forms a dual system.
}


   \begin{lemma}\label{lem:compl}Let $\cC = (\cx, \cz)$ be a CSS code over $\F_q^d$ of dimension $k$, and let $\cx = \subs{W}{X} \oplus \cz^{\perp}$ and $\cz = \subs{W}{Z}  \oplus \cx^{\perp}$ respectively. Then, the canonical bilinear form over $\F_q$ is non-degenerate over $\subs{W}{X} \,\times \,\subs{W}{Z}$. 
 Therefore, there exists a pair of $\F_q$--linear isomorphisms $\phix: \F_{q}^{b} \to W_X$, $\phiz: \F_{q}^{b} \to W_Z$ such that, 
\[
\ip{x}{y}_{\F_q^b} = \ip{\phix(x)}{\phiz(y)}_{\F_q^d}  \; \; \forall\, x,y \in \F_{q^b} \mper
\]
\end{lemma}

\begin{proof} The space $(\F_q^k,\F_q^k)$ with the canonical bilinear form is a dual system with the elementary basis as a dual basis.  Therefore, if $(\subs{W}{X}, \subs{W}{Z} )$ equipped with the canonical bilinear form is a dual system, then~\cref{claim:dual_map} yields a duality-preserving isomorphism as needed. To show this, we only need to prove the non-degeneracy of the canonical form, and we will do so for one component as the argument is symmetric. 

Let $v \in \subs{W}{Z}$ be such that $\ip{u}{v} = 0$ for all $u \in \subs{W}{X} $. Since, $v \in  \subs{W}{Z} \subseteq \cz$, we have $\ip{w}{v} = 0$ for all $w \in \cz^\perp$. But, $\cx = \subs{W}{X} \oplus \cz^\perp$ and therefore, $v \in \cx^\perp$. Since, $\subs{W}{Z} \cap \cx^\perp = \{0\}$, $v$ must be $0$. 
\end{proof}

\paragraph{Changing base field} We now see that one can view a CSS code over $\F_{q^b}$ as a vector space CSS code over $\F_q^b$. To do this we first 
equip $\F_{q^b}$ with the \textit{trace form}.

  \begin{definition}[Trace Map]
	Let $\F_{q^b}$ be a degree $k$-extension of $\F_q$. The \textit{trace map}\footnote{We will drop the subscript as we will not work with multiple extensions.} is defined as 
	\[\mathsf{Tr}_{\,\F_{q^b}/\F_q} :\F_{q^b} \to \F_q, \;\;  x \mapsto x + x^q + x^{q^2} + \cdots + x^{q^b}.\]
{	The trace map  is $\F_q$-linear as for any $a \in \F_q$, $a^q = a$. }
\end{definition} 

It is a well-known fact that the trace map defines a non-degenerate $\F_q$-bilinear map over $\F_{q^b} \times \F_{q^b}$ defined as $\ip{x}{y}_\mathsf{Tr} = \mathsf{Tr} (xy)$. Therefore, $(\F_{q^b} ,\F_{q^b}, \ip{\cdot}{\cdot}_\mathsf{Tr} )$ forms a dual system, and so from~\cref{claim:dual_map}, there exists a pair of maps $\phix, \phi_z: \F_{q^b} \to \F_q^b$, that is duality preserving. We denote by $(\wphix, \wphiz)$, the map obtained on $\F_{q^b}^n$ obtained by applying $(\phix, \phiz)$ on each coordinate.


\begin{lemma}\label{lem:field_downgrade}
	Let $\cC$ be a $[[n,k,\delta n]]_{q^b}$ CSS code over $\F_{q^b}$, and let $(\phix,\phiz)$ be the duality preserving map defined above. Then $\cC' = \parens{\wphix(\cx), \wphiz(\cz)}$ is a $[[n,k, \delta n]]_{q, b}$ vector space CSS code, where  
\end{lemma}
\begin{proof}
To see that it is a CSS code, let $u  = \wphix(u') \in \wphiz(\cz)^\perp$. The vector $u$ can be written this way as $\wphix$ is an isomorphism. Then, for any $v \in \cz$,
\begin{align*}
 \ip{u'}{v} = \ip{\wphiz(u')}{\wphiz(v)} = 0. 	
\end{align*}
Thus, $u' \in \cz^\perp$ and therefore, $\wphiz(\cz)^\perp \subseteq \wphix(\cz^\perp) \subseteq \wphix(\cx)$ which is the CSS condition.

	Since $\cx$ is a $\F_{q^b}$- subspace, $\dim(\phix(\cx))$ over $\F_q$ is $b\cdot \dim(\cx)$. Thus, the dimension of the new (unfolded code) is $\dim(\phix(\cx)) - \dim(\phix(\cz)^\perp) = b\cdot k$ which becomes $k$ after folding. 
	The distance is unchanged as for $v = (v_1,\cdots, v_n) \in \F_{q^b}^n$ the inital weight is, 
	\[\abs{v} = \abs{\{ i \mid v_i\neq 0 \}}= = \abs{\{ i \mid \phix(v_i) \neq 0 \}} = \abs{\phix(v)}.\]
\end{proof}
\subsection{Expander Graphs}
For a bipartite graph $G=(L,R,E)$, let $L$ be the set of left vertices, and $R$ be the set of right vertices. Let $A_G$ denote the $L\times R$ biadjacency matrix, and $\sigma_2(A_G)$ be its second largest singular value.
 


%
\begin{definition}[$(n,d,\lambda)$-expander]
A $d$-regular bipartite graph $G(L,R,E)$ with $|L|=|R|=n$ is said to be an $(n,d,\lambda)$--expander if $\sigma_2(A_G) ~\leq~ \lambda \cdot d$.
\end{definition}

Infinite families of $(n,d,\lambda)$--expanders, with growing $n$ as $d$ and $\lambda$ are constant, can be derived based on double covers of Ramanujan graphs of~\cite{LPS88} as long as $\lambda \geq \frac{2\sqrt{d-1}}{d}$.


\begin{lemma}[Expander Mixing Lemma]\label{lem:eml}
	Given an $(n,d,\lambda)$-expander $G=(L,R,E)$ and functions $f: L \rightarrow \R$ and $g: R \rightarrow \R$, the following well--known property is a simple consequence of definition of $(n,d,\lambda)$--expanders:
	\ifnum\confversion=1
	\small
	\[
		\abs*{\Ex{(\li,\ri) \sim E}{ f(\li) \cdot g(\ri)} - \Ex{\li\sim L}{f(\li)} \Ex{\ri\sim R}{g(\ri)}}
                \leq \lambda \cdot \norm{f}_2 
                \norm{g}_2
	\]
	\normalsize
	\else
	\[\abs[\Big]{\Ex{(\li,\ri) \sim E}{ f(\li) \cdot g(\ri)} - \Ex{\li\sim L}{f(\li)} \cdot \Ex{\ri\sim R}{g(\ri)}}
                \leq \lambda \cdot \norm{f}_2\norm{g}_2 \mper
	\]
	\fi
\end{lemma}


%

%
%
%

%
%


\section{Concatenated Codes and AEL Amplification}
\label{sec:ael}

In this section, we formalize the quantum CSS generalizations of two operations on classical codes -- \textit{code concatenation} and AEL \textit{distance amplification} (\cite{AEL95}). We define it purely in linear algebraic terms but the recent work~\cite{BGG22} also formalizes this using the stabilizer code framework.

\subsection{Concatenation of CSS Codes}

To define a concatenation of CSS codes, one needs a pair of outer code and inner code that are compatible with respect to some parameters. Additionally in the quantum case, we need a pair of duality-preserving maps as defined in~\cref{lem:compl}.


\begin{definition}[Concatenated CSS Codes]
Given the following objects,	
\begin{itemize}
    \item \textbf{Outer Code} --- Let $\cD = (\dx, \dz)$ be a $[[n, k_{\out}, \delta_{\out}]]_{q,b_\out}$ vector space CSS code.
    \item \textbf{Inner Code} --- Let $\cC = (\cx, \cz)$ be a $[[d, k_\inn, \delta_{\inn}]]_{q,b_\inn}$ vector space CSS code that $b_\inn k_\inn = b_\out$. Let $\cx = W_X \oplus \cz^{\perp}$ and $\cz = W_Z \oplus \cx^{\perp}$ respectively (as $\F_q^{db_\inn}$-subspaces). 
    \item \textbf{Duality preserving maps} --- Let $(\phix,\phiz)$ be duality preserving maps as in~\cref{lem:compl} from $\F_{q}^{k}$ to $\subs{W}{X}$ and $\subs{W}{Z}$ respectively. Extend these maps to $(\F_{q}^k)^n$ by applying it to each coordinate {of the folded code}, and call the extended map $(\wphix, \wphiz)$.
\end{itemize}
\vspace{-5pt}One defines the concatenated CSS Code, $\cE = \cD \circ_\varphi \cC$ as $(\ex, \ez)$ with
    \begin{align*}
        \ex &=\wphix( \dx) + \mathbb{F}_{q}^n \otimes \cz^{\perp} \subseteq \F_q^{ndb_\inn}\\
        \ez &= \wphiz(\dz) + \mathbb{F}_{q}^n \otimes \cx^{\perp}\subseteq \F_q^{ndb_\inn}
    \end{align*}
\end{definition}

We now give an explicit description of the dual spaces that will prove that the concatenation operation defines a CSS code and also be useful in proving the distance of the final code. 

%
%

\begin{proposition}\label{prop:concat_dual}
    For the above definition of $(\ex, \ez)$, the dual spaces can be computed as follows,
    \begin{align*}
        \ex^{\perp} &= \wphiz(\dx^{\perp}) + \mathbb{F}_{q}^n \otimes \cx^{\perp} \\
        \ez^{\perp} &= \wphix(\dz^{\perp}) + \mathbb{F}_{q}^n \otimes \cz^{\perp}
    \end{align*}
    Therefore, $\cE$ is an $[[nd, {k_{\out}\cdot k_\inn}]]_{q, b_\inn}$ vector space CSS code.
    {Moreover, if $\cD$ is an LDPC code, and $d b_\inn$ is a constant, then $\cE$ is an LDPC code.}
\end{proposition}

\begin{proof} We prove the first equation as the proofs are symmetric.  Define, \[\mathcal{U}_X := \wphiz(\dx^{\perp}) + \mathbb{F}_{q}^n \otimes \cx^{\perp}.\] 

Since $\wphix$ is an $\F_q$-linear isomorphism $\dim(\wphix(U)) = \dim(U)$. Moreover, $\mathrm{im}(\wphiz)$ and $\cx^\perp$ are disjoint and thus their dimension add up.    Thus,
 \begin{align*}
 	\dim(\mathcal{U}_X) ~&=~ \dim(\dx^\perp) + n\, \dim(\cx^\perp) = nb_\out-\dim(\dx) + n (db_\inn-\dim(\cx)), \\
 	\dim(\ex) ~&=~ \dim(\dx) + n(\dim(\cz^\perp)) = \dim(\dx) + n(\dim(\cx) -k_\inn ) .
 \end{align*}
 
 Therefore, $\dim(\mathcal{U}_X) = ndb_\inn - \dim(\ex) =  \dim(\ex^\perp)$ and it suffices to show that $\,\mathcal{U}_X\subseteq \ex^\perp$.
 
Let $\alpha \in \mathcal{U}_X$ and $\beta \in \ex$. Using the definition of the spaces, we express them as,
\begin{align*}
 \alpha ~&=~ (\phiz(u^{(1)}) + x_1,\,\cdots\,, \phiz(u^{(n)}) + x_n ) \in \mathcal{U}_X \; \text{ where } x_i \in \cx^\perp,\, u =(u^{(1)},\cdots, u^{(n)}) \in\dx^\perp,\\   
 \beta ~&=~ (\phix(v^{(n)}) + z_1,\,\cdots\,, \phix(v^{(n)}) + z_n ) \in \ex  \;\;\, \text{ where } z_i \in \cz^\perp,\,  v = (v^{(1)},\cdots, v^{(n)}) \in\dx.
\end{align*}

Computing the inner product we get four kinds of terms,
\[
\ip{\alpha}{\beta} = \sum_{i=1}^n \parens[\Big]{\ip{\phiz\parens[\big]{u^{(i)}}}{\phix\parens[\big]{v^{(i)}}} + \ip{\phiz\parens[\big]{u^{(i)}}}{z_i} + \ip{x_i}{\phix\parens[\big]{v^{(i)}}} + \ip{x_i}{z_i} }  \]

Each of the last three terms is zero as, by definition, they belong to orthogonal spaces.
We are then left with the first term which can be calculated using the duality-preserving property,
\vspace{-0.5em}
\begin{align*}
   \ip{\alpha}{\beta} ~&=~ \sum_{i=1}^n \angles[\Big]{\phiz\parens[\big]{u^{(i)}}\;, \;\phix\parens[\big]{v^{(i)}}}_{\F_q^{db_\inn}} \\
   ~&=~ \sum_{i=1}^n  \ip{u^{(i)}}{v^{(i)}}_{\F_q^{b_\out}} &&\text{(Duality Preserving)}\\
   ~&=~  \ip{u}{v}_{\F_q^{nb_\out}} = 0 && (u \in \dx^\perp, v \in \dx). 
\end{align*}

This proves that $\mathcal{U}_X = \ex^\perp$. Moreover, if $\dx = V_X \oplus \dz^\perp$, the proof implies that $\ex = \phi_X(V_X) \oplus \ez^\perp$. Recall that for a $[[n,k_{\out}]]_{q,b_\out}$ vector space code, the dimension of $V_X$ as a $\F_q$-subspace is $k_{\out}\cdot b_\out$.
   Thus, the dimension of the CSS code $\cE$ is, \[\dim(\ex) -\dim(\ez^\perp) =  \dim(\phi_X(V_X)) = \dim(V_X) = k_{\out}\cdot b_\out = k_{\out}\cdot k_\inn b_\inn \quad .\]
Folding this $\cE$ into blocks of size $b_\inn$, we get a $[[nd, k_\out \cdot k_\inn]]_{q, b_\inn}$ vector space CSS code.    
The LDPC property follows since the generators of $\ex^\perp$ are comprised of generators of $\dz^\perp$ mapped by $\wphix$, and generators corresponding to $\cz^{\perp}$. The former are sparse if $\cD$ is LDPC, and the latter has weight at most $db_\inn$ which is a constant.
\end{proof}


\subsection{AEL Amplification and Folding CSS Codes}

The distance of the concatenated code $\cE$  can be amplified by the AEL procedure using a $d$-regular bipartite expander, $G = (L,R, E)$.
The graph $G$ is chosen such that the size of $L$ and $R$ match the blocklength of the outer code, and the degree matches the blocklength of the inner code. 

 The AEL procedure is a three-step process --- (i) concatenate the outer code $\cD$ with inner code $\cC$ to obtain $\cE$, (ii) shuffle the symbols of concatenated code via edges on a bipartite expander graph $G$, and (iii) collect $d$-symbols on the right vertices and fold them back to produce the final code, $\cF$. 

In this subsection, we will formally state the AEL procedure and set up some useful notation. We start by restating the definitions of the concatenated codes (and their duals) in a manner that will be convenient when working with AEL, and later, with sum-of-squares proofs.



%
\paragraph{Concatenated Codes and AEL} To simplify notation, will use $\Sigma$ to denote $\F_q^{b_\inn}$ as the concatenated code lies inside the space $(\F_q^{b_{\inn}})^{nd}  = \Sigma^E$.  We view the codewords, $z \in \ex$ (or $\ez$), as an assignment of $\Sigma$-values to the edges. Denote by $z_\li \in \Sigma^d$, the restriction of the vector $z$ to the neighborhood, $N(\li)$, of vertex $\li \in L$. We will use similarly use $z_\ri$, to denote restriction to the neighborhood of vertex $\ri \in R$.  
	The concatenated code $\cE = \cD\circ_\varphi \cC$ and its duals can be explicitly described as,
\begin{align*}
 	\ex ~&=~ \Set*{ x \mid x_{\li} =  \phix(u_\li) + z_\li, \text{ for a unique  } z_\li \in \cz^\perp \text{ and } u \in \dx },\\[2.5pt]
 		\ex^\perp ~&=~ \Set*{ x \mid x_{\li} =  \phix(u_\li) + z_\li, \text{ for a unique  } z_\li \in \cx^\perp \text{ and } u \in \dx^\perp },\\[2.5pt]
 		\ez ~&=~ \Set*{ x \mid x_{\li} =  \phix(u_\li) + z_\li, \text{ for a unique  } z_\li \in \cx^\perp \text{ and } v \in \dz },\\[2.5pt]
 		\ez^\perp ~&=~ \Set*{ x \mid x_{\li} =  \phix(u_\li) + z_\li, \text{ for a unique  } z_\li \in \cz^\perp \text{ and } u \in \dx }.
 \end{align*}

 Uniqueness follows as $\mathrm{im}(\phix) = \subs{W}{X}$ is disjoint from $\cz^\perp$ (and similarly for $\phiz$). The concatenated code $\cE$ is folded using the partitions induced by the neighborhoods of the right vertices. Explicitly, the folded code $\cF = \fold(\cE) =  (\fold(\ex), \fold(\ez))$ is given by,
 \begin{align*}
 	\fx:=~ \fold(\ex) ~&=~ \Set[\Big]{\fold(z_{\ri_1}, \cdots , z_{\ri_n}) \big\vert \, z \in \ex,\; \ri_i \in R} \subseteq \parens*{\Sigma^d}^n  \\
 	\fz:=~ \fold(\ez) ~&=~ \Set[\Big]{\fold(z_{\ri_1}, \cdots , z_{\ri_n})  \big\vert\, z \in \ez,\; \ri_i \in R} \subseteq \parens*{\Sigma^d}^n 
 	\end{align*} 

\begin{proposition}[AEL Procedure]\label{cor:ael_rate}
Let $\cD$ be a $[[n, k_\out]]_{q, b_\out}$ CSS code and $\cC$ be a $[[d, k_\inn]]_{q, b_\inn}$ CSS code.\\[2.5pt]
Then, the AEL code, $\cF = (\fx, \fz)$ defines a $[[n, \frac{k_\out \cdot k_\inn}{d}, \delta_R\cdot n]]_{q, db_\inn}$ CSS code where \[\delta_R\cdot n = \min \braces[\big]{\, \abs[\big]{\braces{i \mid z_{r_i}\neq 0}} \;\big\vert\;  z \in (\ex\setminus \ez^\perp) \cup (\ez\setminus \ex^\perp)}. \] 
\end{proposition}

The key property of the amplified code, $\cF$, is that its distance is significantly better than $\cE$. We will prove a lower bound on $\delta_R$ in the next section (\cref{thm:ael_final}). Before that, we define two notions that we use later: \textit{local inversion maps}, and a couple different distance metrics over $\Sigma^{dn}$.

\paragraph{Local Inversion} The uniqueness of decomposition of the local codeword, $x_\li$, let us define an inverse to the maps $\phix$ and $\phiz$.
\begin{definition}[Local Inversion Maps]\label{def:local_inversion}
Let $\cE = \cD\circ_\varphi \cC$ be the concatenated code as above. Then one defines local inversion maps,
\begin{align*}
 \unconx :\cx \rightarrow \F_q^{b_{\out}},\;\; x_\li ~&=~ \phix(u_\li) + z_\li \mapsto u_\li, \\  
 \unconz :\cz \rightarrow \F_q^{b_{\out}},\;\; x_\li ~&=~ \,\phiz(v_\li) + z_\li \mapsto v_\li.
\end{align*}
	
\end{definition}

\paragraph{Distance metrics for $\cE$} Using the graph structure, we can fold the code $\cE$ using the left or right vertices. Moreover, we can define a define a metric for the set of $\cz^\perp$ cosets which is needed for the quantum notion of distance. 
 \begin{align*}
 	\Delta_L(z,h) ~&=~ \Ex{\ell\sim L}{\indi{ z_\li \neq h_{\li}}},  \\
 	\Delta_{L, \cz^\perp}(z,h) ~&=~ \Ex{\ell\sim L}{\indi{ z_\li \not\in h_{\li} + \cz^\perp }}, \\
 	\Delta_{R}(z,h) ~&=~\Ex{r\sim R}{\indi{ z_{r} \neq h_{r}}}.
 \end{align*}
  Now we can reinterpret AEL procedure as changing the metric from the initial $\Delta_L(\cdot, \cdot)$ on the concatenated code $\cE$ to  $\Delta_R(\cdot, \cdot)$. This change is crucial as this is where the pseudorandom properties of the graph (expansion) come in, and imply that the distance between codewords under the $\Delta_R$ metric is much larger than the initial distance under the $\Delta_L$ metric.
  
  \paragraph{Decoders for AEL}
  As defined in \cref{sum:list_dec_vector_space}, a list decoder up to radius $\tau$ for the code $\cF$ should take as input a string $g\in (\Sigma^d)^R$, and output a list of cosets of $\fz^\perp$ containing 
  \begin{align*}
      \calL(g,\tau) ~&=~ \braces[\big]{ \fold(h)+\fz^\perp \suchthat h\in \ex,\, \Delta_R(g,h) <\tau}\\
      ~&\cong~ \braces[\big]{ h+\ez^\perp \suchthat h\in \ex,\, \Delta_R(g,h) <\tau} .
  \end{align*}

  The second set is merely the unfolded version of the first and it is equivalent to work with either. We will work with the latter to simplify notation.

\section{Distance Proofs}\label{sec:distance}

In this section, we prove that the folded AEL code, $\fold(\cD \circ_\varphi \cC)$ has large (fractional) distance, $\delta_R$, that can be made arbitrarily close to that of the inner code $\delta_\inn$ by picking a good enough expander.

\subsection{Partial Minimizer}
As mentioned before, one of the key objects we will use in the distance proof is a partial minimizer, which does not change the coset, but gets closer to the codeword we are measuring distance from. For a fixed codeword $h \in \ex$, and any vector $z \in \F_q^E$ we define the partial minimizer to be a new vector $\partmin(z,h)$    
\[	\partmin(z,h)_{\li} = \begin{cases}
		z_{\li} \;\; \text { if } \;\;  z_{\li} \not\in\,  h_{\li} + \cz^\perp \\[6pt]
		h_{\li} \;\; \text { if } \;\; z_{\li} \in\,  h_{\li} + \cz^\perp 
	\end{cases}
		\]
We observe that the partial minimizer satisfies two key properties,
\begin{align} 
&\text{ (\small{Coset-preserving}) } \;\;\; &&\partmin(z,h)_{\li} \in\, z_{\li} + \cz^\perp \;\; \forall\, \li \in L  \label{eqn:coset}\\[4pt]
	&\text{ (\small{Monotone}) } \;\;\;	&&\Delta(\partmin(z,h)_\ri, h_\ri) \leq \Delta(z_\ri, h_\ri) \;\; \forall \, \ri \in\, R \label{eqn:monotone}
\end{align}
 
\subsection{Distance proof}

We next prove for any two codewords in $\ex$ that do not share the same coset of $\ez^\perp$, their distance in the $\Delta_R(\cdot, \cdot)$ metric is almost as large as $\delta_{\inn}$.

\begin{lemma}[Distance proof of AEL]\label{lem:zfc_dist}
Let $z$ and $h$ be two non-equivalent codewords in $\ex$, \ie $z \not\in h + \ez^\perp$. Then,
\[\Delta_{R}(z,h) \geq \delta := \delta_\inn - \frac{\lambda}{\Delta_{L, \cz^\perp}(z,h)}  \] 	
 where $\Delta_{L,\cz^\perp} (z,h) = \Ex{\li \in L}{\one \braces[\big]{z_{\li} \not\in  h_{\li} + \cz^\perp}} $. 
\end{lemma}
\begin{proof}
	We will lower bound and upper bound the same quantity $\Delta_E(\partmin(z,h), h)$. Note that by \cref{eqn:coset}, we have $\Delta_{L,\cz^\perp} (z,h) = \Delta_{L,\cz^\perp} (\partmin(z,h),h)$.
	
	\begin{align*}
			\Ex{e\in E}{\partmin(z,h)_e \neq h_e} ~&=~ \Ex{\li\in L}{{\dist{\partmin(z,h)_{\li} , h_{\li} }}} \\
			~&\geq~ \Ex{\li\in L}{\indi{\partmin(z,h)_{\li} \not\in  h_{\li} + \cz^\perp} \cdot \dist{\partmin(z,h)_{\li},h_{\li} }} \\
			~&\geq~ \Ex{\li\in L}{\indi{\partmin(z,h)_{\li} \not\in  h_{\li} + \cz^\perp} \cdot \delta_\inn} \\
			~&=~ \delta_\inn \cdot \Delta_{L,\cz^\perp} (\partmin(z,h),h) \\
			~&=~ \delta_\inn \cdot \Delta_{L,\cz^\perp} (z,h)&& (\text{Using \cref{eqn:coset}}) .
	\end{align*}

	For the upper bound, we will use the expander mixing lemma~(\cref{lem:eml}).
	\begin{align*}
		\Ex{e\in E}{{\indi{\partmin(z,h)_e \neq h_e}}} ~&\leq~ \Ex{(\li, \ri)\, \in\, E}{{ \indi{\partmin(z,h)_{\li} \neq h_{\li} } \cdot \indi{\partmin(z,h)_{\ri} \neq h_{\ri} }}} \\[4pt]
		~&\leq~ \Ex{\li, \ri}{{ \indi{\partmin(z,h)_{\li} \neq h_{\li} } \cdot \indi{\partmin(z,h)_{\ri} \neq h_{\ri} }}} + \lambda &&\text{(\cref{lem:eml})}\\[4pt]
		~&=~\Delta_{L} (\partmin(z,h),h)\cdot \Delta_R(\partmin(z,h),h) + \lambda  \\[4pt]
		~&=~ \Delta_{L,\cz^\perp} (\partmin(z,h),h)\cdot \Delta_R(\partmin(z,h),h) + \lambda &&\text{(Definition of $\partmin$)} \\[4pt]
		~&=~ \Delta_{L,\cz^\perp} (z,h)\cdot \Delta_R(\partmin(z,h),h) + \lambda &&\text{(\cref{eqn:coset})} \\[4pt]
		~&\leq~ \Delta_{L,\cz^\perp} (z,h)\cdot \Delta_R(z,h) + \lambda &&\text{(\cref{eqn:monotone})}
\end{align*}
	 Comparing the two sides, we get  
	\begin{align*}
		\Delta_{R}(z,h)\cdot \Delta_{L, \cz^\perp}(z,h) + \lambda ~\geq~ \delta_\inn \cdot \Delta_{L, \cz^\perp}(z,h)
	\end{align*}
	Since $z\not\in h+ \ez^\perp$, there exists at least one vertex $\li \in L$ such that $z_\li \not \in h_\li + \cz^\perp$ and thus, $\Delta_{L,\cz^\perp} (z,h) > 0$.  Dividing by it gives the result.	
\end{proof}

We now deduce that using AEL machinery amplifies the distance of the base outer code. 

\begin{theorem}[AEL distance]\label{thm:ael_final}
Let $\cD$ be a $[[n, k_\out, \delta_\out\cdot n]]_{q, b_\out}$ vector space CSS code and $\cC$ be a $[[d, k_\inn, \delta_{\inn}\cdot d]]_{q, b_\inn}$ vector space CSS code.
Let $\cF$ be the AEL code obtained by using an $(n,d,\lambda)$-expander. Then, $\cF$ is a $[[n, \frac{k_\out \cdot k_\inn}{d}, \delta_R\cdot n]]_{q, db_\inn}$ vector space CSS code where $\delta_R \geq \delta_{\inn} - \frac{\lambda}{\delta_{\out}}$.
\end{theorem}
\begin{proof}
	\cref{cor:ael_rate} gives the dimension of the code. To compute the distance, we observe that  $\delta_R = \min_v \disR{v}{0}$ where the min is over non-trivial codewords. The result now follows from the above distance bound~\cref{lem:zfc_dist}. 
\end{proof}

As mentioned before, we will need an SoS version of the above distance proof for our decoding algorithm. 
Before describing this SoS distance proof and how the above defined notion of partial minimizer assists it, we will need some additional preliminaries about the Sum-of-Squares hierarchy. We discuss these preliminaries in \cref{sec:sos_prelims}, and return to the distance proof in \cref{sec:sos_distance}.

\section{Additional Preliminaries: Sum-of-Squares Hierarchy}\label{sec:sos_prelims}

The sum-of-squares hierarchy of semidefinite programs (SDPs) provides a family of increasingly
powerful convex relaxations for several optimization problems. 
Each ``level" $t$ of the hierarchy is parametrized by a set of constraints corresponding to
polynomials of degree at most $t$ in the optimization variables. While the relaxations in the
hierarchy can be viewed as  semidefinite programs of size $n^{O(t)}$ \cite{BS14, FKP19}, 
it is often convenient to view the solution as a linear operator, called the ``pseudoexpectation" operator.
%

%
%
%
\paragraph{Pseudoexpectations}
Let $t$ be a positive even integer and fix an alphabet $\Sigma$ of size $s$. Let $\zee = \{Z_{i,j}\}_{i\in[m],j\in[s]}$ be a collection of variables and $\R[\zee]^{\leq t}$ be the vector space of polynomials of degree at most $t$ in the variables $\zee$ (including the constants).

\begin{definition}[Constrained Pseudoexpectations]\label{def:constraints_on_sos}
Let $\calS = \inbraces{f_1 = 0, \ldots, f_m = 0, g_1 \geq 0, \ldots, g_r \geq 0}$ be a system of
polynomial constraints in $\zee$, with each polynomial in $\calS$ of degree at most $t$. We say $\tildeEx{\cdot}$ is a pseudoexpectation operator of SoS-degree $t$, over the variables $\zee$  respecting $\calS$, if it is a linear operator $ \tildeEx{\cdot}: \R[\zee]^{\leq t} \rightarrow \R$ such that:
	\begin{enumerate}
	\item $\tildeEx{1} = 1$.
    \item $\tildeEx{p^2} \geq 0$ if $p$ is a polynomial in $\zee = \{Z_{i,j}\}_{i\in [m],j\in [s]}$ of degree $\leq t/2$.
	\item $\tildeEx{p \cdot f_i} = 0$,  $\forall\, i \in [m]$ and $\forall\, p$ such that $\deg(p \cdot f_i) \leq t$.
	\item $\tildeEx{p^2 \cdot \prod_{i \in S} g_i} \geq 0$, $\forall\, S \subseteq [r]$ and $\forall\, p$ such that $\deg(p^2\cdot \prod_{i \in S} g_i) \leq t$.
	\end{enumerate}
\end{definition}
~


%


Let $\mu$ be a distribution over the set of assignments, $\Sigma^m$.  Define the following collection of random variables,  \[\zee = \braces[\big]{ \, Z_{i,j}  = \indi{ i \mapsto j}\, \mid \, i\in[m],\, j\in[s] } .\] Then, setting $\tildeEx{p(\zee)} = \Ex{\mu}{p(\zee)} $ for any polynomial $p(\cdot)$ defines an (unconstrained) pseudoexpectation operator. However, the reverse is not true when $t < m$, and there can be degree-$t$ pseudoexpectations that do not correspond to any genuine distribution, $\mu$. Therefore, the set of all pseudoexpectations should be seen as a relaxation for the set of all possible distributions over such assignments. The upshot of this relaxation is that it is possible to optimize over the set. Under certain conditions on the bit-complexity of solutions~\cite{OD16, RW17:sos}, one can optimize over the set of degree-$t$ pseudoexpectations in time $m^{O(t)}$ via SDPs.

\paragraph{Local constraints and local functions.}
Any constraint that involves at most $k$ variables from $\zee$, with $k\leq t$, can be written as a degree-$k$ polynomial, and such constraints may be enforced into the SoS solution.
%
%
In particular, we will always consider the following canonical constraints on the variables $\zee$.
\ifnum\confversion=1
\begin{align*}
&Z_{i,j}^2 = Z_{i,j},\ \forall i\in[m],j\in[s] \\
\text{and} \quad &\sum_j Z_{i,j} = 1,\ \forall i\in[m] \mper
\end{align*}
\else
\[
Z_{i,j}^2 = Z_{i,j},\ \forall \,i\in[m],j\in[s] 
\quad \text{and} \quad 
\sum_j Z_{i,j} = 1,\ \forall\, i\in[m] \mper
\]
\fi
%
%
%
We will also consider additional constraints and corresponding polynomials, defined by ``local" functions. For any $f\in \Sigma^m$ and $M\sub [m]$, we use $f_M$ to denote the restriction $f|_M$, and $f_i$ to denote $f_{\{i\}}$ for convenience.
\begin{definition}[$k$-local function]
	A function $\mu: \Sigma^m \rightarrow \R$ is called $k$-local if there is a set $M\subseteq [m]$ of size $k$ such that $\mu(f)$ only depends on $\inbraces{f(i)}_{i\in M}$, or equivalently, $\mu(f)$ only depends on $f|_M$.
	
	If $\mu$ is $k$-local, we abuse notation to also use $\mu: \Sigma^M \rightarrow \R$ with $\mu(\alpha) = \mu(f)$ for any $f$ such that $f|_M=\alpha$. It will be clear from the input to the function $\mu$ whether we are using $\mu$ as a function on $\Sigma^m$ or $\Sigma^M$.
\end{definition}

Let $\mu:\Sigma^m\rightarrow \R$ be a $k$-local function that depends on coordinates $M\subseteq [m]$ with $|M|=k$. Then $\mu$ can be written as a degree-$k$ polynomial $p_{\mu}$ in $\zee$:
\[
	p_{\mu}(\zee) = \sum_{\alpha \in \Sigma^M} \parens[\Big]{\mu(\alpha) \cdot\prod_{i\in M} Z_{i,\alpha_i}}
\]


With some abuse of notation, we let $\mu(\zee)$ denote $p_{\mu}(\zee)$. We will use such $k$-local
functions inside $\tildeEx{\cdot}$ freely without worrying about their polynomial
representation. For example, $\tildeEx{ \indi{\zee_{i} \neq j}}$ denotes $\tildeEx{ 1- Z_{i,j}}$. 
%
The notion of $k$-local functions can also be extended from real valued functions to vector valued functions in a straightforward way.

\begin{definition}[Vector-valued Local functions]
A function $\mu: \Sigma^m \rightarrow \R^N$ is $k$-local if the $N$ real valued functions corresponding to the $N$ coordinates are all $k$-local. Note that these different coordinate functions may depend on different sets of variables, as long as the number is at most $k$ for each of the functions.
\end{definition}
\paragraph{Local distribution view of SoS}

It will be convenient to use a shorthand for the function $\indi{\zee_A = \alpha}$, and we will use $\zee_{A,\alpha}$. Likewise, we use $\zee_{i,j}$ as a shorthand for the function $\indi{\zee_i = j}$. That is, henceforth,
\ifnum\confversion=1
\begin{align*}
	&\tildeEx{\zee_{A,\alpha}} = \tildeEx{\indi{\zee_A = \alpha}} = \tildeEx{ \prod_{a\in A}
                             Z_{a,\alpha_a}}
                             \\
	\text{and } \quad & \tildeEx{\zee_{i,j}} = \tildeEx{\indi{\zee_i = j}} = \tildeEx{ Z_{i,j}}.
\end{align*}
\else
\begin{align*}
	\tildeEx{\zee_{A,\alpha}} ~=~ \tildeEx{\indi{\zee_A = \alpha}} ~=~ \tildeE\brackets[\Big]{\prod_{a\in A}
                      Z_{a,\alpha_a}
                                    }
\qquad \text{and} \qquad
	\tildeEx{\zee_{i,j}} ~=~ \tildeEx{\indi{\zee_i = j}} = \tildeEx{ Z_{i,j}}
\end{align*}
\fi


Note that for any $A \subseteq [m]$ with $\abs*{A} = k \leq t/2$,
\ifnum\confversion=1
\begin{gather*}
	\sum_{ \alpha \in \Sigma^{k}} \tildeEx{\zee_{A,\alpha}} = 
 \tildeEx{ \prod_{a\in A} \inparen{ \sum_{j\in \Sigma} Z_{a,j}} } = 1
\\
	\tildeEx{\zee_{A,\alpha}} = \tildeEx{ \prod_{a\in A} Z_{a,\alpha_a}} = \tildeEx{ \prod_{a\in
            A} Z^2_{a,\alpha_a}} \geq 0 \mper
\end{gather*}
\else
\[
	\sum_{ \alpha \in \Sigma^{k}} \tildeEx{\zee_{A,\alpha}} = 
 \tildeE\brackets[\bigg]{ \prod_{a\in A} \parens[\bigg]{\sum_{j\in \Sigma} Z_{a,j} } } = 1
\qquad \text{and} \qquad
	\tildeEx{\zee_{A,\alpha}} = \tildeE\brackets[\bigg]{ \prod_{a\in A} Z_{a,\alpha_a}} = \tildeE\brackets[\bigg]{ \prod_{a\in
            A} Z^2_{a,\alpha_a}} \geq 0 \mper
\]
\fi
Thus, the values $\inbraces{\tildeEx{\zee_{A, \alpha}}}_{\alpha\in \Sigma^A}$ define a distribution
over $\Sigma^k$. We call this the local distribution for $\zee_A$, or simply for $A$.
%
%
Let $\mu: \Sigma^m \rightarrow\R$ be a $k$-local function for $k\leq t/2$, depending on $M \subseteq
[m]$. Then,
\ifnum\confversion=1
\begin{align*}
	\tildeEx{\mu(\zee)} 
~=~& \tildeEx{\sum_{\alpha\in \Sigma^M} \inparen{\mu(\alpha) \cdot\prod_{i\in M} Z_{i,\alpha_i}}}\\
~=~& \sum_{\alpha\in \Sigma^M} \mu(\alpha) \cdot \tildeEx{\prod_{i\in M} Z_{i,\alpha_i}}\\
~=~& \sum_{\alpha\in \Sigma^M} \mu(\alpha) \cdot \tildeEx{\zee_{M,\alpha}}
\end{align*}
\else
\begin{align*}
	\tildeEx{\mu(\zee)} 
~=~ \tildeE\brackets[\bigg]{\sum_{\alpha\in \Sigma^M} \parens[\bigg]{\mu(\alpha) \cdot\prod_{i\in M} Z_{i,\alpha_i}}}
~=~ \sum_{\alpha\in \Sigma^M} \mu(\alpha) \cdot \tildeE\brackets[\Big]{\prod_{i\in M} Z_{i,\alpha_i}}
~=~ \sum_{\alpha\in \Sigma^M} \mu(\alpha) \cdot \tildeEx{\zee_{M,\alpha}}
\end{align*}
\fi
That is, $\tildeEx{\mu(\zee)}$ can be seen as the expected value of the function $\mu$ under the local distribution for $M$.

\begin{claim}\label{claim:sos_domination}
	Let $\tildeEx{\cdot}$ be a degree-$t$ pseudoexpectation. For $k \leq t/2$, let $\mu_1,\mu_2$
        be two $k$-local functions on $\Sigma^m$, depending on the same set of coordinates $M$, and
        $\mu_1(\alpha) \leq \mu_2(\alpha) ~~\forall \alpha \in \Sigma^M$. Then $\tildeEx{\mu_1(\zee)} \leq \tildeEx{\mu_2(\zee)}$.
%
\end{claim}

\begin{proof}
Let $\calD_M$ be the local distribution induced by $\tildeEx{\cdot}$ for $\zee_M$. Then
$\tildeEx{\mu_1(\zee)} = \Ex{\alpha \sim \calD_M}{\mu_1(\alpha)}$, and $\tildeEx{\mu_2(\zee)} =
\Ex{\alpha\sim \calD_M}{\mu_2(\alpha)}$, which implies $\tildeEx{\mu_1(\zee)} \leq \tildeEx{\mu_2(\zee)}$.
%
\end{proof}
The previous claim allows us to replace any local function inside $\tildeEx{\cdot}$ by another local function that dominates it. We will make extensive use of this fact.
\vspace{-5 pt}
\paragraph{Covariance for SoS solutions}
Given two sets $S,T \sub [m]$ with $|S|,|T|\leq k/4$, we can define the covariance between indicator random variables of $\zee_S$ and $\zee_T$ taking values $\alpha$ and $\beta$ respectively, according to the local distribution over $S \cup T$. This is formalized in the next definition.
\begin{definition}
Let $\tildeEx{\cdot}$ be a pseudodistribution operator of SoS-degree-$t$, and $S,T$ are two sets of
size at most $t/4$, and $\alpha\in \Sigma^S$, $\beta\in \Sigma^T$, we define the pseudo-covariance and
pseudo-variance,
\ifnum\confversion=1
\small
\begin{gather*}
\tildecov(\zee_{S,\alpha},\zee_{T,\beta}) 
= \tildeEx{ \zee_{S,\alpha} \cdot \zee_{T,\beta} } - \tildeEx{\zee_{S,\alpha}} \tildeEx{\zee_{T,\beta}} \\	
\tildeVar{\zee_{S,\alpha}} ~=~ \tildecov(\zee_{S,\alpha},\zee_{S,\alpha})
\end{gather*}
\normalsize
\else
\begin{align*}
\tildecov(\zee_{S,\alpha},\zee_{T,\beta}) 
~&=~ \tildeEx{ \zee_{S,\alpha} \cdot \zee_{T,\beta} } - \tildeEx{\zee_{S,\alpha}} \,\tildeEx{\zee_{T,\beta}}\\ 		
\tildeVar{\zee_{S,\alpha}} ~&=~ \tildecov(\zee_{S,\alpha},\zee_{S,\alpha})
\end{align*}
\fi
The above definition is extended to pseudo-covariance and pseudo-variance for pairs of sets $S,T$, 
as the sum of absolute value of pseudo-covariance for all pairs $\alpha,\beta$ :
\ifnum\confversion=1
\begin{gather*}
\tildecov(\zee_S,\zee_T) 
~=~ \sum_{\alpha\in \Sigma^S \atop \beta\in \Sigma^T} \abs*{ \tildecov(\zee_{S,\alpha},\zee_{T,\beta}) } \\
\tildeVar{\zee_S} ~=~ \sum_{\alpha\in \Sigma^S} \abs*{ \tildeVar{\zee_{S,\alpha} } }
\end{gather*}
\else
\begin{align*}
\tildecov(\zee_S,\zee_T) 
~&=~ \sum_{\alpha\in \Sigma^S, \beta\in \Sigma^T} \abs*{ \tildecov(\zee_{S,\alpha},\zee_{T,\beta}) }\\[3pt]
\tildeVar{\zee_S} ~&=~ \sum_{\alpha\in \Sigma^S} \abs*{ \tildeVar{\zee_{S,\alpha} } }
\end{align*}
\fi
\end{definition}

%
We will need the fact that $\tildeVar{\zee_S}$ is bounded above by 1, since,
\ifnum\confversion=1
\begin{align*}
\tildeVar{\zee_S} 
&~=~ \sum_{\alpha} \abs*{\tildeVar{\zee_{S,\alpha}}} \\
&~=~ \sum_{\alpha}\inparen{
          \tildeEx{\zee_{S,\alpha}^2} - \tildeEx{\zee_{S,\alpha}}^2} \\
&~\leq~ \sum_{\alpha} \tildeEx{\zee_{S,\alpha}^2} \\
&~=~ \sum_{\alpha} \tildeEx{\zee_{S,\alpha}} 
~=~ 1.
\end{align*}
\else
\[
\tildeVar{\zee_S} 
~=~ \sum_{\alpha} \abs*{\tildeVar{\zee_{S,\alpha}}} 
~=~ \sum_{\alpha}\inparen{
          \tildeEx{\zee_{S,\alpha}^2} - \tildeEx{\zee_{S,\alpha}}^2} 
~\leq~ \sum_{\alpha} \tildeEx{\zee_{S,\alpha}^2} 
~=~ \sum_{\alpha} \tildeEx{\zee_{S,\alpha}} 
~=~ 1.
\]
\fi

%
\vspace{-5 pt}
\paragraph{Conditioning SoS solutions.}
We will also make use of conditioned pseudoexpectation operators, which may be defined in a way
similar to usual conditioning for true expectation operators, as long as the event we condition on
is local. 
The conditioned SoS solution is of a smaller degree but continues to respect the constraints that the original solution respects.

\begin{definition}[Conditioned SoS Solution] Let $F \subseteq \Sigma^m$ be subset (to be thought of as an event) such that $\one_F:\Sigma^m \rightarrow \{0,1\}$ is a $k$-local function. Then for every $t>2k$, we can condition a pseudoexpectation operator of SoS-degree $t$ on $F$ to obtain a new conditioned pseudoexpectation operator $\condPE{\cdot}{F}$ of SoS-degree $t-2k$, as long as $\tildeEx{\one^2_F(\zee)}>0$. The conditioned SoS solution is given by
\[
	\condPE{ p(\zee)}{F(\zee) } \defeq \frac{\tildeEx{p(\zee) \cdot \one^2_{F}(\zee)}}{\tildeEx{\one^2_{F}(\zee)}}
\]
where $p$ is any polynomial of degree at most $t-2k$.
\end{definition}

We can also define pseudocovariances and pseudo-variances for the conditioned SoS solutions.
\begin{definition}[Pseudocovariance]
	Let $F\sub \Sigma^m$ be an event such that $\one_F$ is $k$-local, and let $\tildeEx{\cdot}$ be a pseudoexpectation operator of degree $t$, with $t>2k$. Let $S,T$ be two sets of size at most $\frac{t-2k}{2}$ each. Then the pseudocovariance between $\zee_{S,\alpha}$ and $\zee_{T,\beta}$ for the solution conditioned on event $F$ is defined as,
\ifnum\confversion=1
\begin{multline*}
\tildecov(\zee_{S,\alpha},\zee_{T,\beta} \vert F) = \\
\tildeEx{ \zee_{S,\alpha} \zee_{T,\beta} \vert F} - \tildeEx{\zee_{S,\alpha} \vert F} \tildeEx{\zee_{T,\beta} \vert F}
\end{multline*}
\else
\begin{align*}
\tildecov(\zee_{S,\alpha},\zee_{T,\beta} \vert F) 
~=~ \tildeEx{ \zee_{S,\alpha} \zee_{T,\beta} \vert F} - \tildeEx{\zee_{S,\alpha} \vert F} ~ \tildeEx{\zee_{T,\beta} \vert F}
\end{align*}
\fi
\end{definition}

We also define the pseudocovariance between $\zee_{S,\alpha}$ and $\zee_{T,\beta}$ after
conditioning on a random assignment for some $\zee_V$ with $V\sub [m]$. 
Note that here the random assignment for $\zee_V$ is chosen according to the local distribution for
the set $V$.

\begin{definition}[Pseudocovariance for conditioned pseudoexpectation operators]
\ifnum\confversion=1
\begin{multline*}
\tildecov(\zee_{S,\alpha},\zee_{T,\beta} \vert \zee_V) 
= \\ \sum_{\gamma \in \Sigma^V}   \tildecov(\zee_{S,\alpha},\zee_{T,\beta} \vert \zee_V = \gamma) \cdot \tildeEx{\zee_{V,\gamma}}
	\end{multline*}
\else
\begin{align*}
\tildecov(\zee_{S,\alpha},\zee_{T,\beta} \vert \zee_V) 
~=~ \sum_{\gamma \in \Sigma^V}   \tildecov(\zee_{S,\alpha},\zee_{T,\beta} \vert \zee_V = \gamma) \cdot \tildeEx{\zee_{V,\gamma}}
\end{align*}
\fi
\end{definition}

And we likewise define $\tildeVar{\zee_{S,\alpha} \vert \zee_V}$, $\tildecov(\zee_S, \zee_T \vert \zee_V)$ and $\tildeVar{\zee_S \vert \zee_V}$.

\subsection{SoS relaxations for codes}

\paragraph{Pseudocodewords for AEL}

Let $G(L,R,E)$ be the bipartite $(n,d,\lambda)$-expander on which the AEL code is defined, and let $\cC = (\cx,\cz)$ be the inner code. The SoS variables will be $\zee = \{Z_{e,j}\}_{e\in E,j\in \Sigma}$, where $\Sigma = \F_q^{b_{\inn}}$. 

The subsets of $E$ we consider will all be vertex neighborhoods from the left or right. To simplify notation, we use the convention that if $\li \in L$, then $\zee_{\li}$ denotes $\zee_{N(\li)}$, where $N(\li)$ is the set of all edges incident on the left vertex $\li$. Likewise, if $\ri \in R$, $\zee_{\ri}$ denotes $\zee_{N(\ri)}$, where $N(\ri)$ denotes the set of all edges incident on the right vertex $\ri$.

We also extend this shorthand to include subsets of $L$ and $R$ so that if $S \sub L$ ($R$ respectively), then $\zee_S$ denotes $\zee_{N(S)}$ where $N(S)$ is the set of all edges in $E$ incident on left (right respectively) vertices in $S$.

\begin{definition}[$X$-Pseudocodewords]
	For $t\geq 2d$, we define a degree-$t$ AEL pseuocodeword to be a degree-$t$ pseudoexpectation operator $\tildeEx{\cdot}$ on $\zee$ respecting the following constraints:
	\begin{align*}
		\forall\, \li \in L,\quad \zee_{\li} \in \cx
	\end{align*}
\end{definition}
We now define a generalized version of the notion of distance between a pseudocodeword $\tildeEx{\cdot}$, and a true codeword $h$.  

\begin{definition}[Distance from a pseudocodeword]
We define three types of distances between a pseudocodeword $\tildeEx{\cdot}$ of SoS-degree $t\geq 2d$ and a (unfolded) codeword $h \in \Sigma^E$ are defined as
\begin{align*}
 	\distL{\tildeEx{\cdot}} ~&=~ \Ex{\li\sim L}{\tildeEx{\indi{ \zee_\li \neq h_{\li}}}}\\
\distLperp{\tildeEx{\cdot}} ~&=~ \Ex{\li \sim L}{\tildeEx{\indi{ \zee_\li \not\in h_{\li} + \cz^\perp }}}\\
 \distR{\tildeEx{\cdot}} ~&=~\Ex{r\sim R}{\tildeEx{\indi{ \zee_{r} \neq h_{r}}}}
 \end{align*}
 \end{definition}
 
Next, we define a non-convex property for a pseudocodeword that allows us to prove a distance property similar to \cref{lem:zfc_dist} for pseudocodewords.

\begin{definition}\label{def:eta_good}
	A pseudocodeword is $\eta$-good if,
	\[
		\Ex{\li,\ri}{ \tildeCov{ \zee_{\li}}{\zee_{\ri} } } \leq \eta
	\]
\end{definition}
%
%
The following lemma from \cite{JST23} captures the main property of $\eta$-good pseudocodoewords that we will be using. The proof of this lemma can be found in \cref{app:correlation_rounding} where it is restated as \cref{lem:using_eta_good}.
\begin{lemma}\label{lem:eta_good}
	Let $\{X_{\li}\}_{\li\in L}$ and $\{Y_{\ri}\}_{\ri\in R}$ be two collections of $d$-local functions on $\Sigma^E$ such that for every $\li\in L$, $X_{\li}(f)$ only depends on $f_{\li}$ and for every $\ri\in R$, $Y_{\ri}(f)$ only depends on $f_{\ri}$. Then, for an $\eta$-good pseudocodeword $\tildeEx{\cdot}$,
	\[
		\Ex{\li,\ri}{\tildeEx{X_{\li}(\zee) \cdot Y_{\ri}(\zee)}} ~\leq~ \Ex{\li,\ri}{\tildeEx{X_{\li}(\zee)} \cdot \tildeEx{Y_{\ri}(\zee)}} + \eta\, \parens[\Big]{ \max_{\li} \norm{X_{\li}}_{\infty} } \parens[\Big]{ \max_{\ri} \norm{Y_{\ri}}_{\infty}}
	\]
\end{lemma}

Next, we show that a variant of the expander mixing lemma can be adapted to work inside the SoS proof system by exploiting the non-negativity of sum-of-squares of low-degree polynomials. The next lemma is again borrowed from \cite{JST23}, although it has appeared in different forms before in the literature, and we include a proof in \cref{app:eml_sos} where it is restated as \cref{lem:eml_pexp_appendix}.

\begin{lemma}[EML for pseudoexpectations]\label{lem:eml_pseudoexpectation}
	Let $\{X_{\li}\}_{\li\in L}$ and $\{Y_{\ri}\}_{\ri\in R}$ be two collections of $d$-local functions on $\Sigma^E$ such that for every $\li\in L$, $X_{\li}(f)$ only depends on $f_\li$ and for every $\ri\in R$, $Y_{\ri}(f)$ only depends on $f_\ri$.
Then for a $\lambda$-spectral expander, we have
	\[
		\abs*{\Ex{\li\sim \ri}{\tildeEx{X_{\li}(\zee) \cdot Y_{\ri}(\zee) }} -
                  \Ex{\li,\ri}{\tildeEx{X_{\li}(\zee) \cdot Y_{\ri}(\zee)}}} \leq \lambda
                \sqrt{\Ex{\li}{\tildeEx{X_{\li}(\zee)^2}}} \cdot \sqrt{\Ex{\ri}
                  {\tildeEx{Y_{\ri}(\zee)^2}}} \mper
	\]
\end{lemma}


\section{SoS Proof of Distance}\label{sec:sos_distance}



We start by extending the definition of partial minimizer to be a vector-valued local function so that can be used inside pseudoexpectations. For a fixed codeword $h$, we define
\[
	\partmin(\zee,h)_{\li} = \begin{cases}
		\zee_{\li} \;\; \text { if } \;\;  \zee_{\li} \not\in\,  h_{\li} + \cz^\perp \\[6pt]
		h_{\li} \;\; \text { if } \;\; \zee_{\li} \in\,  h_{\li} + \cz^\perp 
	\end{cases}
\]
More explicitly, $\partmin(\zee,h)_{\li}$ is shorthand for $\indi{\zee_{\li} \not\in h_{\li} + \cz^{\perp}} \cdot \zee_{\li} + \indi{\zee_{\li} \in h_{\li} + \cz^{\perp}} \cdot h_{\li}$.
%
\begin{align} 
	&\text{ (\small{Coset-preserving}) } \;\;\; &&\partmin(\zee,h)_{\li} \in\, \zee_{\li} + \cz^\perp \iff \unconx(\partmin(\zee,h)_{\li}) = \unconx(\zee_{\li}) \;\; \forall\, l \in L  \label{eqn:coset_sos}\\[4pt]
	&\text{ (\small{Monotone}) } \;\;\;	&&\indi{ \partmin(\zee, h)_r \neq h_r} ~\leq~ \indi{\zee_r \neq h_r} \;\; \forall \, r \in\, R \label{eqn:monotone_sos}
\end{align}
Recall that $\unconx$ is the local inversion map defined in \cref{def:local_inversion}. We emphasize that the partial minimizer is importantly defined in a way that the monotonicity property is true \emph{locally}.

Next, we prove a distance bound analogous to \cref{lem:zfc_dist} for pseudocodewords that satisfy the $\eta$-good property from \cref{def:eta_good}. This is the key statement we need to make the framework from \cite{JST23} applicable to the quantum AEL setting.

\begin{lemma}[Distance proof]\label{lem:sos_ael_distance}
Let $\tildeEx{\cdot}$ be an $\eta$-good pseudocodeword and $h \in \Sigma^E$ be such that $\fold(h)$ is a codeword in $ \fx$ and $\distLperp{\tildeEx{\cdot}} > 0$. Then,
\[\distR{\tildeEx{\cdot}} ~\geq~ \delta_\inn - \frac{\lambda + 
\eta}{\distLperp{\tildeEx{\cdot}}}  \mper\] 	
\end{lemma}
\begin{proof}
	We will closely follow the proof of~\cref{lem:zfc_dist} and similarly, will lower bound and upper bound the quantity $\Ex{e\in E}{~\tildeEx{~\indi{\partmin(\zee, h)_e \neq h_e}}}$. 
\begin{align*}
			\Ex{e\in E}{~\tildeEx{~\indi{\partmin(\zee, h)_e \neq h_e}}} ~&=~ \Ex{\li\in L}{\tildeEx{\dist{\partmin(\zee, h)_{\li} , h_{\li} }}} \\
			~&\geq~ \Ex{\li\in L}{\tildeEx{\indi{\partmin(\zee, h)_{\li} \not\in h_{\li} + \cz^\perp} \cdot \dist{\partmin(\zee, h)_{\li} , h_{\li} }}} \\
			~&\geq~ \Ex{\li\in L}{\tildeEx{\indi{\partmin(\zee, h)_{\li} \not\in h_{\li} + \cz^\perp} \cdot \delta_{\inn} }} \\
			~&=~ \delta_{\inn} \cdot \Ex{\li\in L}{\tildeEx{\indi{\partmin(\zee, h)_{\li} \not\in h_{\li} + \cz^\perp} }} \\
			~&=~ \delta_{\inn} \cdot \Ex{\li\in L}{\tildeEx{\indi{\zee_{\li} \not\in h_{\li} + \cz^\perp} }} && (\text{Using \cref{eqn:coset_sos}}) \\
			~&=~ \delta_{\inn} \cdot \distLperp{\tildeEx{{\cdot}}}
	\end{align*}

For the upper bound, we will use the expander mixing lemma~(\cref{lem:eml_pseudoexpectation}).
	\begin{align*}
		\Ex{e\in E}{~\tildeEx{~\indi{\partmin(\zee, h)_e \neq h_e}}} ~&\leq~ \Ex{\li \sim \ri}{~\tildeEx{~\indi{\partmin(\zee,h)_{\li} \neq h_{\li} } \cdot \indi{\partmin(\zee,h)_{\ri} \neq h_{\ri} }}} \\[4pt]
		~&\leq~ \Ex{\li, \ri}{\tildeEx{ \indi{\partmin(\zee, h)_{\li} \neq h_{\li} } \cdot \indi{\partmin(\zee, h)_{\ri} \neq h_{\ri} }}} + \lambda &&\text{( \cref{lem:eml_pseudoexpectation} )}\\[4pt]
		~&\leq~ \Ex{\li, \ri}{\tildeEx{ \indi{\partmin(\zee, h)_{\li} \neq h_{\li} }} \cdot \tildeEx{\indi{\partmin(\zee, h)_{\ri} \neq h_{\ri} }}} + \lambda +\eta &&\text{( \cref{lem:eta_good} )}\\[4pt]
		~&=~ \Ex{\li}{\tildeEx{ \indi{\partmin(\zee, h)_{\li} \neq h_{\li} }}} \cdot \Ex{\ri}{\tildeEx{\indi{\partmin(\zee, h)_{\ri} \neq h_{\ri} }}} + \lambda +\eta \\[4pt]
		~&=~ \Ex{\li}{\tildeEx{ \indi{\zee_{\li} \not\in h_{\li} + \cz^{\perp}}}} \cdot \Ex{\ri}{\tildeEx{\indi{\partmin(\zee, h)_{\ri} \neq h_{\ri} }}} + \lambda +\eta &&\text{ (Definition of $\partmin$)} \\[4pt]
		~&\leq~ \Ex{\li}{\tildeEx{ \indi{\zee_{\li} \not\in h_{\li} + \cz^{\perp}}}} \cdot \Ex{\ri}{\tildeEx{\indi{\zee_{\ri} \neq h_{\ri} }}} + \lambda +\eta &&\text{( \cref{eqn:monotone_sos} )} \\[4pt]
		~&=~\distLperp{\tildeEx{\cdot}} \cdot \distR{\tildeEx{\cdot}}  + \lambda + \eta\;.
\end{align*}

Comparing the two sides, we get  
	\begin{align*}
		\distLperp{\tildeEx{\cdot}} \cdot \distR{\tildeEx{\cdot}} + \lambda + \eta  ~\geq~ \delta_{\inn} \cdot \distLperp{\tildeEx{{\cdot}}} \; .
	\end{align*}
	Dividing by $\distLperp{\tildeEx{{\cdot}}}$ gives us the result.
	\end{proof}

\section{List Decoding Algorithm}\label{sec:decoding}

In this section, we show how to use the machinery of \cite{JST23} to develop a list decoding algorithm for the quantum codes obtained by AEL amplification. The two main components of this machinery are the (algorithmic) covering lemma and the (SoS) distance proof. The former lets us discover an object "containing" the list, while the latter allows the list elements to be extracted from this object.

We start with a subsection on covering lemmas (\cref{sec:covering}), which we then combine with the SoS distance proof from \cref{sec:sos_distance} in \cref{sec:combine}. We show how to derandomize the algorithm in \cref{sec:derandomization}, followed by the application to QLDPC codes in \cref{sec:near_mds}.

We stress that all the results in this section follow the outline from \cite{JST23}, except with minor changes to accommodate the differences in the quantum setting, such as the change from $\Delta_{L}(\cdot , \cdot)$ to $\Delta_{L,\cz^{\perp}}(\cdot , \cdot)$. We choose to write down the formal details due to these minor changes, but defer some proofs to the appendix if they are exactly same as in the classical case.

\subsection{Covering Lemma}\label{sec:covering}

As we saw in \cref{sec:sos_prelims}, pseudocodewords can be seen as a relaxation of the set of distributions over codewords. For list decoding, we would like to ensure that a (pseudo-)distribution does not completely ignore some of the codewords in the list. In particular, we would like a way to avoid point distributions over single codewords. 

The key idea in \cite{JST23}, going back further to \cite{AJQST20}, is that maximizing certain entropy measure indeed allows us to obtain a (pseudo-)distribution that is close to \emph{every} codeword in the list. We slightly modify the statements from \cite{JST23} to focus only on the large alphabet case. This part of the algorithm is very analytic in that it works on embeddings of codewords into Euclidean spaces, and does not need to care about the underlying code structure. 

In particular, there is no difference between classical and quantum codes in the embedding and covering lemma, but the proof of Johnson bound using the covering lemma takes into account the different notion of distance, so that we have distance only between pairs of codewords in $\fx$ that belong to different cosets of $\fz^\perp$.

For any $f\in \Sigma^E$, recall that $f_r \in  \Sigma^d$ is the restriction of $f$ to the edge-neighborhood $N(\ri)$ of a vertex $r \in R$. Let us denote the size of alphabet $\Sigma = \F_q^{b_{\inn}}$ by $s \defeq q^{b_{\inn}}$.
	
	\begin{definition}[Embedding]
	Let $\R^{s^d}$ be spanned by formal orthonormal basis vectors $\Set{e_w\,\mid \,w\in \Sigma^d}$. Define an embedding $\embed : \Sigma^E \to \parens{\R^{s^d}}^R \cong \R^{s^d\cdot n}$, as $\embed (f)_r = e_{f_r}$ for each $r\in R$. We denote the image of a subspace, $U$ under this map $\embed(U)$. The Euclidean inner product for this embedding captures the $\disR{\cdot}{\cdot}$ distance as,
	\[ \disR{f_1}{f_2} = 1 - \ip{\embed(f_1)}{\embed(f_2)}.
	\]
	\end{definition}
	We will use this embedding to prove the following lemma, as well as its algorithmic version \cref{lem:algo_covering}. Since the next lemma is only used to prove the existential Johnson bound, and is completely same as the classical case in \cite{JST23}, we defer its proof to \cref{app:covering}.
\begin{lemma}[Covering Lemma]\label{lem:covering}
	Let $g\in \Sigma^E$ and $\alpha \in (0,1)$.
	There exists a distribution $\calD$ over $\ex$ with support size at most $s^d \cdot n + 1$ such that for any $h \in \ex$ such that $\distR{g} < 1-\alpha$, the distribution $\calD$ satisfies
	\[
		\Ex{f\sim \calD}{\distR{f}} < 1-\alpha^2 \mper
	\]
\end{lemma}

Using the covering lemma above, we now prove the Johnson bound for quantum CSS codes obtained by AEL amplification. While this proof can be easily adapted for general CSS codes, we choose to write it for the specific case of quantum AEL codes of \cite{BGG22} in preparation for its algorithmic version in \cref{lem:algo_covering}.
\begin{lemma}[Alphabet-free Johnson bound]
	Let $\cF = (\fx,\fz)$ be a CSS code of distance $\delta$ obtained by AEL amplification, so that $\fx = \fold(\ex)$ and $\fz=\fold(\ez)$. Then for any $g\in \Sigma^E$, 
	\begin{equation*}
		\abs{\calL( g, 1-\sqrt{1-\delta} ) } ~\leq~ s^d \cdot n +1
	\end{equation*}
\end{lemma}

\begin{proof}
	We use \cref{lem:covering} for $g$ to obtain a distribution $\calD$ over $\ex$ of support size at most $s^d \cdot n + 1$. This distribution has the property that for any $h \in \ex$ such that $\distR{g} < 1-\sqrt{1-\delta}$, it holds that
	\begin{equation}\label{eqn:covering}
		\Ex{f \sim \calD}{\distR{f}} ~<~ \delta
	\end{equation}
	Fix a coset (of $\ez^\perp$) in the list $\calL(g , 1-\sqrt{1-\delta} )$ and let $h$ be the codeword from this coset that is nearest to $g$ in $\Delta_R(\cdot ,\cdot)$ metric, so that $\distR{g} < 1-\sqrt{1-\delta}$. We claim that the support of $\calD$ must contain a member of the coset $h+\ez^{\perp}$. If not, every $f \in \supp(\calD)$ satisfies $\distR{f} \geq \delta$ due to the distance of the code $\cF$, and so,
	\[
		\Ex{f\sim \calD}{\distR{f}} ~\geq~ \delta
	\]
which contradicts \cref{eqn:covering}. Since the total support of $\calD$ is bounded by $s^d \cdot n + 1$, the number of cosets in the list is also bounded by $s^d\cdot n +1$.
\end{proof}

While the covering lemma is sufficient to conclude an upper bound on the list size up to the Johnson bound, it is not clear how to \emph{efficiently} find the distribution $\calD$ promised in the covering lemma. However, that distribution is shown to exist as the optimizer of certain convex function over the set of distributions over codewords. The next lemma shows that optimizing for the same convex function over \emph{pseudocodewords} also gives similar covering properties, but has the advantage that such an optimal degree-$t$ pseudocodeword can be found via SDPs in time $n^{\calO(t)}$.

\begin{lemma}[Algorithmic Covering Lemma]\label{lem:algo_covering}
	Let $g\in \Sigma^E$, $\alpha \in (0,1)$ and $\eps \in (0,1-\alpha)$. 
	There exists a pseudocodeword $\tildeEx{\cdot}$ of degree $t\geq d$ such that for every $h \in \ex$ such that $\distR{g} < 1-\alpha - \eps$, it holds that
	\[
		\distR{\tildeEx{\cdot}} < 1-\alpha^2 - 2 \alpha \eps \mper
	\]
	Moreover, such a pseudocodeword can be found in time $n^{\calO(t)}$.
\end{lemma}
\begin{proof}
	We claim that a pseudocodeword of degree $t$ that minimizes $\norm{\tildeEx{\embed(\zee)}}^2 = \ip{\tildeEx{\embed(\zee)}}{\tildeEx{\embed(\zee)}}$ among all degree-$t$ pseudocodewords that satisfy $\Delta_R(\tildeEx{\cdot},g) < 1-\alpha - \eps$ has the required property above.
%
	We need to argue that for every $h \in \ex$ such that $\distR{g} < 1-\alpha - \eps$, it holds that
	\begin{equation}\label{eqn:pexp_agreement}
		\ip{\tildeEx{\embed(\zee)}}{\embed(h)} > \alpha^2 + 2 \alpha \eps \mper
	\end{equation}

	This suffices to prove the lemma as 
	\begin{align*}
		\disR{\tildeEx{\cdot}}{h} &= \Ex{r\in R}{\tildeEx{\indi{\zee_r \neq h_r}}} \\
		&= \Ex{r\in R}{\tildeEx{1 - \indi{\zee_r = h_r}}} \\
		&= 1 - \Ex{r\in R}{\tildeEx{\indi{\zee_r = h_r}}} \\
		&= 1 - \Ex{r\in R}{\tildeEx{\ip{\embed(\zee_r)}{\embed(h_r)}}} \\
		&= 1- \ip{\tildeEx{\embed(\zee)}}{\embed(h)}
	\end{align*}
	
	We also note the following lower bound on $\Psi\parens[\Big]{\tildeEx{\cdot}}$:
	\begin{align}
		& \angles[\big]{\tildeEx{\embed(\zee)}\, , \,\tildeEx{\embed(\zee)}}\cdot \ip{\embed(g)}{\embed(g)} ~\geq~ \angles[\big]{\tildeEx{\embed(\zee)} \, , \, \embed(g)}^2 ~>~ (\alpha + \eps)^2 \\
		\implies \quad & \Psi\parens[\Big]{\tildeEx{\cdot}} ~=~ \angles[\big]{\tildeEx{\embed(\zee)}\, , \,\tildeEx{\embed(\zee)}} ~>~ (\alpha + \eps)^2 ~>~ \alpha^2 + 2\alpha \eps \label{eqn:lower_bound_entropy}
	\end{align}
	
	Suppose \cref{eqn:pexp_agreement} does not hold. Then there exists an $h\in \ex$ such that $\distR{g} < 1-\alpha - \eps$ (which is equivalent to $\ip{\embed(g)}{\embed(h)} > \alpha + \eps$), but
	\begin{align*}
		\ip{\tildeEx{\embed(\zee)}}{\embed(h)} \leq \alpha^2 + 2 \alpha \eps \mper
	\end{align*}
	
	For a $\tee \in [0,1]$ to be chosen later, consider the degree-$t$ pseudocodeword $\dupPE{\cdot}$ which can be evaluated for any (vector-valued) $t$-local function $\mu$ as
	\[
		\dupPE{\mu(\zee)} = (1-\tee)\cdot \tildeEx{\mu(\zee)} + \tee \cdot \mu(h) \mper
	\]
	
	We will reach a contradiction to the optimality of $\tildeEx{\cdot}$ by showing that $\Psi\parens*{\dupPE{\cdot}} < \Psi\parens*{\tildeEx{\cdot}}$.
	
	Since every integral codeword like $h \in \ex$ is also a degree-$t$ pseucodeword, this convex combination is also a degree-$t$ pseudocodeword, and it is also easy to see that $\ip{\dupPE{\embed(\zee)}}{\embed(g)} ~>~ \alpha + \eps$. 

	\begin{align*}
		\Psi\parens*{\dupPE{\cdot}}  &= \ip{\dupPE{\chiTan(\zee)}}{\dupPE{\chiTan(\zee)}} \\
		~&=~ \ip{(1-\tee) \cdot \tildeEx{\chiTan(\zee)} + \tee\cdot \chiTan(h)}{(1-\tee)\cdot \tildeEx{\chiTan(\zee)}+\tee \cdot \chiTan(h)} \\
		~&=~ (1-\tee)^2 \cdot \ip{\tildeEx{\chiTan(\zee)}}{\tildeEx{\chiTan(\zee)}} +2\cdot \tee(1-\tee)\cdot \ip{\tildeEx{\chiTan(\zee)}}{\chiTan(h)} + \tee^2\cdot \ip{\chiTan(h)}{\chiTan(h)}\\
		~&\leq~ (1-\tee)^2\cdot \Psi\parens*{\tildeEx{\cdot}} +2\cdot \tee(1-\tee)\cdot \parens{\alpha^2 + 2\alpha \eps} +\tee^2 \;.
	\end{align*}

The optimal $\tee$ is given by
\[
	\tee^* ~=~ \frac{\Psi\parens*{\tildeEx{\cdot}} - \parens{\alpha^2 + 2\alpha \eps}}{\Psi\parens*{\tildeEx{\cdot}}-2\parens{\alpha^2 + 2\alpha \eps}+1}
\]

Since $\Psi\parens*{\tildeEx{\cdot}} > \alpha^2 + 2\alpha \eps$ from \cref{eqn:lower_bound_entropy}, we get optimal $\tee^* >0$, which implies $\Psi\parens*{\dupPE{\cdot}} < \Psi\parens*{\tildeEx{\cdot}}$.
\end{proof}


\subsection{Putting things together}\label{sec:combine}
Next, we combine the algorithmic covering lemma \cref{lem:algo_covering} and the distance proof \cref{lem:sos_ael_distance} to give a list decoding algorithm. However, note that the distance proof is only applicable to $\eta$-good pseudocodewords, and the pseudocodeword found by \cref{lem:algo_covering} need not have this property.

To remedy this, we start with a lemma that says that the $\eta$-good property can be obtained by random conditioning. This technique is common in algorithmic applications of Sum-of-Squares, and first appeared in \cite{BRS11}. It was adapted for decoding in \cite{JST23}. Since the proof is exactly the same as in \cite{JST23}, we defer the proof to \cref{app:correlation_rounding}.

\begin{lemma}[Correlation Rounding]\label{lem:conditioning_eta_good}
	Let $\eta>0$, and $\tildeEx{\cdot}$ be a pseudocodeword of degree at least $2d\left(\frac{s^{3d}}{\eta^2}+1\right)$. Then there exists an integer $u^* \leq s^{3d}/\eta^2 $ such that 
	\[
		\Ex{r_1,r_2,\cdots,r_{u^*}}{\Ex{\li,\ri}{\tildecov\brackets*{\zee_{\li},\zee_{\ri} \vert \zee_{r_1},\zee_{r_2},\cdots,\zee_{r_{u^*}}}}} \leq \eta \mper
	\]
	where $r_1,r_2,\cdots r_{u^*}$ are independent randomly chosen vertices of $R$.
\end{lemma}
That is, conditioning on the neighborhoods of $u^*$ many randomly chosen vertices in $R$ renders the pseudocodeword $\eta$-good in expectation. This can be seen as a form of "sampling" from the pseudodistribution. Sampling from a true distribution of codewords can be seen as conditioning on all $n$ vertex neighborhoods in $R$, and a sampled codeword is always integral, which implies it is $\eta$-good with $\eta=0$. \cref{lem:conditioning_eta_good} shows that there is some non-convex property we can derive even with just constantly many conditionings, and that this argument is also valid for pseudodistributions that do not have high enough moments to support sampling.

We now get to the main technical theorem of this work that gives an algorithm to list decode codes constructed via quantum AEL amplification of unique decodable base codes. We first present a randomized version that outputs each codeword in the list $\calL$ with some small but constant probability. We can repeat this algorithm to output a list of cosets that contains the true list of cosets with probability arbitrarily close to 1. In the next subsection, we will derandomize this algorithm, but will end up with much bigger list sizes since pruning the list can be hard for quantum codes.

\begin{theorem}[List Decoding AEL amplification]\label{thm:list_decoding_ael}
	Let $(\fx, \fz)$ be the code obtained by applying quantum AEL amplification to the outer code $(\dx,\dz)$ and inner code $(\cx,\cz)$ of distance $\delta_{\inn}$ using an $(n,d,\lambda)$--expander graph. The inner code $(\cx,\cz)$ is over an alphabet $\Sigma$ of size $s$, so that $(\fx,\fz)$ is over alphabet $\Sigma^d$ of size $s^d$.

Suppose the code $(\dx, \dz)$ can be unique-decoded from radius $\delta_{\dec}$ in time $\calT(n)$. Assume that $\lambda < \delta_{\dec}$. 
	
	Then for any $\gamma > 0$, there exists an algorithm based on $s^{\calO(d)}/\eps^4$ levels of the SoS hierarchy that given $g\in \Sigma^E$, runs in time $\ln(1/\gamma) \cdot \brackets*{n^{s^{\calO(d)}/\eps^4} + \calT(n)}$ and produces a list $\calL'$ of cosets that contains the list of cosets $\calL \defeq \calL \parens*{g, \calJ\parens*{\delta_{\inn} - \frac{\lambda}{\delta_{\dec}}} - \eps }$ with probability at least $1-\gamma$. The size of $\calL'$ is at most $\widetilde{\Omega}_{s,d} \parens*{ \frac{\ln(1/\gamma)}{\eps^6} }$.
\end{theorem}

\begin{proof}
	The decoding algorithm is presented as \cref{algo:ael-decoding}. 
	Recall that $\calL = \calL \parens*{g, \calJ\parens*{\delta_{\inn} - \frac{\lambda}{\delta_{\dec}}} - \eps }$ is the list of all cosets in $\ex/\ez^{\perp}$ that intersect the Hamming ball $\calB\parens*{g,\calJ\parens*{\delta_{\inn} - \frac{\lambda}{\delta_{\dec}}} - \eps}$.


\begin{figure}[!ht]
\begin{algorithm}{List Decoding}{$g \in \Sigma^E$, $\gamma\in (0,1)$}{List of cosets $\calL' \subseteq \ex / \ez^{\perp}$ that contains $\calL = \calL \parens*{g, \calJ\parens*{\delta_{\inn} - \frac{\lambda}{\delta_{\dec}}} - \eps }$}\label{algo:ael-decoding}
\begin{itemize}
\item Pick $\eta = \frac{\eps^2 \delta_{\dec}}{16\delta_{\inn}}$ and $M = \frac{\ln(1 / \gamma)}{p \ln p}$, where $p = \frac{ \parens*{\delta_{\dec}}^2 \cdot \eps^6}{2^{12} \cdot s^{3d} \cdot \delta_{\inn}^4}$.
\item Use \cref{lem:algo_covering} to obtain a pseudocodeword of degree $t \geq 2d(\frac{s^{3d}}{\eta^2} + 1)$ such that for every $h\in \ex$ which satisfies $\distR{g} < \calJ\parens*{\delta_{\inn} - \frac{\lambda}{\delta_{\dec}}} - \eps$, it holds that
\[
	\distR{\tildeEx{\cdot}} ~<~ \delta_{\inn} - \frac{\lambda}{\delta_{\dec}} - \eps
\]
%
%
%
\item Initialize $\calL' = \emptyset$.
\item Repeat $M$ times:
\begin{enumerate}[(i)]
	\item Choose $u$ uniformly at random from $\{1,2,\cdots , \frac{s^{3d}}{\eta^2}\}$.
	\item Choose a random subset $U \subseteq R$ of size $u$, and let $N(U)$ denote the edge neighborhood of $U$. Sample a random assignment $\sigma$ for $N(U)$ using the local distribution for this set of edges. That is, $\sigma$ is chosen with probability $\tildeEx{\indi{\zee_U = \sigma}}$.
	\item Condition $\tildeEx{\cdot}$ on $\zee_U = \sigma$, and let $\dupPE{\cdot} = \condPE{\cdot}{\zee_U = \sigma}$ be the conditioned pseudocodeword of degree $2d\parens*{\frac{s^{3d}}{\eta^2}+ 1} - 2\cdot d \cdot u \geq 2d$.
	\item Generate $y \in (\F_{q}^{b_{\out}})^L$ by independently sampling $\dupPE{\phi_X^{-1} (\zee_{\li})}$. That is, for $w \in \F_q^{b_{\out}}$, the probability that $y_{\li} = w$ is $\dupPE{\indi{\phi_X^{-1}(\zee_{\li}) = w}}$.
	\item Call the $X$-decoder of $(\dx,\dz)$ on $y$. If a coset $\dxz + \dz^{\perp}$ is found close to $y$,
	then add to $\calL'$ the coset $\wphix(\dxz) + \ez^{\perp}$.
\end{enumerate}
\item Prune $\calL'$ to only include one representative per coset, via Gaussian elimination.
\item Return $\calL'$.
\end{itemize}
\vspace{5pt}
\end{algorithm}
\end{figure}

We now argue that the probability of all the cosets in $\calL$ being included in $\calL'$ is at least $1 - \gamma$. Fix such a coset in $\calL$ and let $h \in \Sigma^E$ be the nearest codeword to $g$ from this coset $h+\ez^\perp$, so that $\distR{g} < \calJ\parens*{\delta_{\inn} - \frac{\lambda}{\delta_{\dec}}} - \eps$. \cref{lem:algo_covering} implies that $\tildeEx{\cdot}$ satisfies
\begin{align}\label{eqn:agreement_before_conditioning}
	\distR{\tildeEx{\cdot}} ~<~ \delta_{\inn} - \frac{\lambda}{\delta_{\dec}} - \eps\;.
\end{align}
To be able to use the distance proof of \cref{lem:sos_ael_distance}, we need the pseudocodeword to be $\eta$-good. \cref{lem:conditioning_eta_good} shows that this can be obtained by random conditioning. In particular, there exists a $u^*\in \{1,2,\cdots ,\frac{s^{3d}}{\eta^2} \}$ such that,
\begin{equation}\label{eq:conditioning_eta_good}
		\Ex{r_1,r_2,\cdots,r_{u^*}}{\Ex{\li,\ri}{\tildecov\brackets*{\zee_{\li},\zee_{\ri} \vert \zee_{r_1},\zee_{r_2},\cdots,\zee_{r_{u^*}}}}} ~\leq~ \eta\;.
\end{equation}
Note that we picked $\eta = \frac{\eps^2 \delta_{\dec}}{16\delta_{\inn}}$ in \cref{algo:ael-decoding}, so that $u^*$ is bounded by a constant independent of $n$.

Suppose $u$ is chosen in step (i) of \cref{algo:ael-decoding} to be $u^*$, which happens with probability at least $\frac{\eta^2}{s^{3d}}$. Rewriting \cref{eq:conditioning_eta_good} with $U$ to denote the randomly chosen set $\{r_1,r_2,\cdots ,r_u\}$ and $N(U)\sub E$ to denote the set of all edges incident on $U \sub R$, we get
\begin{align*}
	\Ex{\substack{U\subseteq R \\ |U| = u}}{\Ex{\li,\ri}{\tildecov\brackets*{\zee_{\li},\zee_{\ri} \vert \zee_{N(U)}}}} ~\leq~ \eta, \\
\Ex{\substack{U\subseteq R, |U| = u \\ \sigma \sim \zee_{N(U)}}}{\Ex{\li,\ri}{\tildecov\brackets*{\zee_{\li},\zee_{\ri} \vert \zee_{N(U)} = \sigma}}} ~\leq~ \eta,
\end{align*}
where $\sigma \sim \zee_{N(U)}$ is used to denote that $\sigma \in \Sigma^{N(U)}$ is sampled according to the local distribution induced on $N(U)$ by $\tildeEx{\cdot}$.

That is, on average, we end up with an $\eta$-good pseudocodeword. We actually picked $\eta$ to be much smaller than the bound on average covariance we will be needing, so that the probability (over conditionings) of obtaining a weaker low covariance becomes very close to 1. This is needed to be able to take a union bound with some other low-probability events we will see soon. A simple application of Markov's inequality shows that the probability of obtaining an $\frac{\eps\delta_{\dec}}{4}$-good pseudocodeword is at least $1-\frac{\eps}{4\delta_{\inn}}$.
\begin{align}\label{eqn:eta_good}
\Pr{U,\sigma}{\Ex{\li,\ri}{\tildecov\brackets*{\zee_{\li},\zee_{\ri} \vert \zee_U = \sigma}} > \frac{\eps \delta_{\dec}}{4}} ~\leq~ \frac{4\eta}{\eps \delta_{\dec}} ~\leq~ \frac{\eps}{4\delta_{\inn}}\;.
\end{align}

Therefore, we started with a pseudocodeword that is close to $h$ (\cref{eqn:agreement_before_conditioning}), and then condition it to make it $\frac{\eps\delta_{\dec}}{4}$-good. We must also argue that this conditioned pseudocodeword is still close to $h$, at least with some probability. This probability cannot be made too large, and this is why we needed to ensure that the low average correlation property holds with probability close to 1, so that both of these hold simultaneously with some positive probability. To do this, we use the law of total expectation and another application of Markov's inequality,
\begin{align}
	\Ex{U,\sigma}{\distR{\condPE{\cdot}{\zee_U = \sigma}}} ~&=~ \Delta_R(\tildeEx{\cdot},h)\;, \\
	~&<~ \delta_{\inn} - \frac{\lambda}{\delta_{\dec}} - \eps\;, && (\text{Using }\cref{eqn:agreement_before_conditioning}) \\
	\implies \Pr{U,\sigma}{\distR{\condPE{\cdot}{\zee_U = \sigma}} \leq \delta_{\inn} - \frac{\lambda}{\delta_{\dec}} - \frac{\eps}{2}} ~&\geq~ \frac{\eps/2}{\delta_{\inn} - \frac{\lambda}{\delta_{\dec}} - \frac{\eps}{2}} ~\geq~ \frac{\eps}{2\delta_{\inn}}\;.\label{eqn:agreement_after_conditioning}
\end{align}

Using a union bound over \cref{eqn:eta_good} and \cref{eqn:agreement_after_conditioning}, we get
\begin{align}\label{eqn:union_bound}
	\Pr{V,\sigma}{\Ex{\li,\ri}{\tildecov\brackets*{\zee_{\li},\zee_{\ri} \vert \zee_U = \sigma}} ~\leq~ \frac{\eps \delta_{\dec}}{4} \text{ and } \distR{\condPE{\cdot}{\zee_U = \sigma}} \leq \delta_{\inn} - \frac{\lambda}{\delta_{\dec}} - \frac{\eps}{2}} ~\geq~ \frac{\eps}{4\delta_{\inn}}\;.
\end{align}

Suppose such a conditioning pair $(U,\sigma)$ is chosen in \cref{algo:ael-decoding}, and this happens with probability at least $\frac{\eps}{4\delta_{\inn}}$.

We define $\dupPE{\cdot} = \condPE{\cdot}{\zee_U = \sigma}$ as in \cref{algo:ael-decoding}, and let us call the corresponding covariance operator as $\dupCov\brackets*{\cdot}$. Rewriting \cref{eqn:union_bound}, we see that $\dupPE{\cdot}$ satisfies the following two properties:
\begin{gather*}
	\Ex{\li,\ri}{\dupCov\brackets*{\zee_{\li},\zee_{\ri}}} ~\leq~ \frac{\eps \delta_{\dec}}{4} \;, \\
	\distR{\dupPE{\cdot}} ~\leq~ \delta_{\inn} - \frac{\lambda}{\delta_{\dec}} - \frac{\eps}{2}\;.
\end{gather*}

Using the SoS distance proof for AEL from \cref{lem:sos_ael_distance} for $\dupPE{\cdot}$, we can use the above upper bound on $\distR{\dupPE{\cdot}}$ to deduce an upper bound on $\distLperp{\dupPE{\cdot}}$.
\begin{align*}
	&\distR{\dupPE{\cdot}} ~\geq~ \delta_{\inn} - \frac{\lambda + \eps \delta_{\dec} / 4}{\distLperp{\dupPE{\cdot}}} \\
	\implies \quad & \quad \frac{\lambda}{\delta_{\dec}} + \frac{\eps}{2} ~\leq~ \frac{\lambda + \eps \delta_{\dec}/4}{\distLperp{\dupPE{\cdot}}} \\
	\implies \quad & \distLperp{\dupPE{\cdot}} ~\leq~ \delta_{\dec} - \frac{\eps \delta_{\dec}}{4\parens*{\frac{\lambda}{\delta_{\dec}} + \frac{\eps}{2}}} ~\leq~ \delta_{\dec} - \frac{\eps \delta_{\dec}}{4 \delta_{\inn}}\;.
\end{align*}

We next wish to show that $y\in (\F_q^{b_{\out}})^L$ obtained by rounding $\tildeEx{\cdot}$ in step (iv) of \cref{algo:ael-decoding} can be used to find the coset $h+\ez^{\perp}$. Let $\dxh$ be the codeword in $\dx \sub (\F_q^{b_{\out}})^L$ corresponding to $h \in \ex$. In other words, $\dxh_{\li} = \unconx(h_{\li})$. 

Let $\dist{y, \dxh}$ denote the normalized Hamming distance between $y$ and $\dxh$, viewed as strings of length $n$ over the alphabet $\F_q^{b_{\out}}$. On average, $y$ satisfies
\begin{align*}
	\Ex{y}{ \dist{y, \dxh}} &= \Ex{y}{ \Ex{\li \in L}{\indi{y_{\li} \neq \dxh_{\li}}}} \\
	&= \Ex{\li \in L}{~\Ex{y}{ \indi{y_{\li} \neq \dxh_{\li}}}} \\
	&= \Ex{\li \in L}{~\Ex{y_{\li}}{ \indi{y_{\li} \neq \dxh_{\li}}}} \\
	&= \Ex{\li \in L}{~ \sum_{w}{ \dupPE{\indi{\unconx(\zee_{\li}) = w}} \indi{w \neq \dxh_{\li}}}} \\
	&= \Ex{\li \in L}{~ \dupPE{\indi{\unconx(\zee_{\li}) \neq \dxh_{\li}}}} \\
	&= \Ex{\li \in L}{~ \dupPE{\indi{\unconx(\zee_{\li}) \neq \unconx(h_{\li})}}} \\
	&= \Ex{\li \in L}{~ \dupPE{\indi{\zee_{\li} \not\in h_{\li} + \cz^{\perp}}}}\\
	&= \distLperp{\dupPE{\cdot}} ~\leq~ \delta_{\dec} - \frac{\eps \delta_{\dec}}{4 \delta_{\inn}} \mper
\end{align*}

Using Markov's inequality,
\begin{align*}
	\Pr{y}{ \dist{y, \dxh} \leq \delta_{\dec}} \geq \frac{\eps}{4\delta_{\inn}} \mper
\end{align*}

Suppose a $y$ is found in \cref{algo:ael-decoding} such that $\dist{y, \dxh} \leq \delta_{\dec}$, which happens with probability at least $\frac{\eps}{4\delta_{\inn}}$. Then the $X$-decoder of the $(\dx,\dz)$ code must return the coset $\dxh+ \dz^{\perp}$. Let $\dxz$ be a coset representative of $\dxh+\dz^\perp$ returned by the $X$-decoder of $(\dx,\dz)$. Using $\ez^{\perp} = \wphix(\dz^{\perp}) + \mathbb{F}_{q}^n \otimes \cz^{\perp}$ from \cref{prop:concat_dual}, 
\begin{align*}
	\dxh - \dxz ~&\in~ \dz^{\perp}\;, \\
	\wphix(\dxh - \dxz) ~&\in~ \ez^\perp\;, \\
	\wphix(\dxh) - \wphix(\dxz) ~&\in~ \ez^\perp\;, \\
	h - \wphix(\dxz) ~&\in~ \ez^\perp + \F_q^n \otimes \cz^\perp = \ez^\perp\;.
\end{align*}

Therefore, \cref{algo:ael-decoding} adds the coset $\wphix(\dxz)+\ez^\perp = h+\ez^\perp$ to $\calL'$. In conclusion, if the following three events happen, the coset $h + \ez^{\perp}$ is added to the list $\calL'$.
\begin{enumerate}
\item $u=u^*$ is chosen, which happens with probability at least $\frac{\eta^2}{s^{3d}}$.
\item The pair $(U,\sigma)$ to condition on is chosen such that \cref{eqn:union_bound} holds. Conditioned on previous event, this happens with probability at least $\frac{\eps}{4\delta_{\inn}}$.
\item A $y$ is generated so that $\dist{y,\dxh} \leq \delta_{\dec}$. Conditioned on above two events, this happens with probability at least $\frac{\eps}{4\delta_{\inn}}$.
\end{enumerate}

Therefore, in any iteration, the coset $h + \ez^{\perp}$ is added to the list $\calL'$ with probability at least,
\[
	p ~=~ \frac{\eta^2}{s^{3d}} \cdot \frac{\eps}{4\delta_{\inn}} \cdot \frac{\eps}{4\delta_{\inn}} ~=~ \frac{ \parens*{\delta_{\dec}}^2 \cdot \eps^6}{4096 \cdot s^{3d} \cdot \delta_{\inn}^4} ~=~ \Omega_{s,d,\delta_{\dec}} (\eps^6) \mper
\]

Note that this immediately implies an upper bound of $1/p$ on the list size, although we can get better list sizes (combinatorially) by appealing to the covering lemma and approximate \Caratheodory theorem.

Finally, we show that with enough repetitions, $\calL'$ must contain the entire list $\calL$ with high probability. The probability that a coset in $\calL$ does not get added to $\calL'$ in $M$ iterations is
\[
	(1-p)^M ~\leq~ e^{-p \cdot M}
\]

With a union bound over the entire list, which is of size at most $1/p$, the probability that any coset in $\calL$ is not present in $\calL'$ is at most
\[
	\frac{1}{p} \cdot e^{-p \cdot M} ~=~ e^{-p M \ln p} ~\leq~ \gamma \;\text{ if } M = \frac{\ln(1 / \gamma)}{p \ln p}. \qedhere
\]
\end{proof}

\subsection{Derandomization}\label{sec:derandomization}

In this section, we show how to derandomize the list decoding \cref{algo:ael-decoding}. In this process, the size of the list $\calL'$ output by the algorithm may be significantly larger, and in particular may grow polynomially with $n$. In contrast, the list size for the probabilistic version was bounded by a constant independent of $n$. Once again, this derandomization is very similar to the one used by \cite{JST23}, and we only include it for completeness.

Suppose we have the same setting as in the last section, and we borrow notation from there to avoid defining the same objects twice. We wish to show that \cref{algo:ael-decoding} can be modified to deterministically output a list $\calL'$ of polynomial size that contains the true list $\calL$. There are three steps where \cref{algo:ael-decoding} makes random choices, and we will address each of them separately:
\begin{enumerate}
\item We try out all $\frac{s^{3d}}{\eta^2}$ possible values of $u$ instead of picking one randomly.
\item We try all pairs $(U,\sigma)$ to condition on. These are at most $n^{\frac{s^{3d}}{\eta^2}} \cdot s^{d\frac{s^{3d}}{\eta^2}}$ many in number, which causes a multiplicative overhead of $ (s^dn)^{\frac{s^{3d}}{\eta^2}}$ in running time and list size.
\item Finally, for independently rounded $y$, we use a standard technique called threshold rounding, also used in \cite{JST23}, to try out all possible $y$ that could be generated. This is described in the next \cref{lem:derandomized_decoding_from_distributions}.
\end{enumerate}
\begin{lemma}\label{lem:derandomized_decoding_from_distributions}
	Let $(\dx, \dz)$ be the outer code defined over alphabet $\F_q^{b_{\out}}$, which is unique decodable from distance $\delta_{\dec}\leq \delta/2$ in time $\calT(n)$. Given a pseudocodeword $\tildeEx{\cdot}$ with the property
	\[
		\distLperp{\tildeEx{\cdot}} \leq \delta_{\dec} \;\;\text{ for some $h\in \ex$,}
	\]
we can find the coset $h+\ez^{\perp}$ in time $\calO(q^{b_{\out}} n)\cdot\calT(n)$ with a deterministic algorithm.
%
	\begin{proof}
	We will use $\tildeEx{\cdot}$ to obtain $|L|=n$ distributions over $\F_q^{b_{\out}}$, say $\{ \calW_{\li} ~|~ \li \in L\}$. We will not use the vector space structure of $\F_q^{b_{\out}}$ here and only view it as a set of size $t\defeq q^{b_{\out}}$. We identify elements of $\F_q^{b_{\out}}$ with those of $[t]$ under an arbitrary ordering, and henceforth, we will view each $\calW_{\li}$ as a distribution over $[t]$. If $w_{\li j}$ denotes the weight of distribution $\calW_{\li}$ on the symbol $j\in [t]$,
	\[
		w_{\li j} = \tildeEx{\indi{\phi_X^{-1}(\zee_{\li}) = j}}
	\]
	We round $\{ \calW_{\li} ~|~ \li \in L\}$ to an $h'\in (\F_q^{b_{\out}})^L$ vertex by vertex. Choose a random $\theta \in [0,1]$. For each $\li \in L$, $h'_{\li}$ is defined to be the smallest $j\in [t]$ such that  $\theta < w_{\li 1}+w_{\li 2} + \cdots  + w_{\li j}$. 
	
	As before, let $\dxh \in \dx$ be the codeword corresponding to $h\in \ex$, that is, $\dxh_{\li} = \unconx(h_{\li})$. Therefore, the probability of $h'_{\li} = j$ is precisely $w_{\li j}$, as before. These events across $\li \in L$ are no longer independent however. But linearity of expectation can be used to show that the key claim from the previous section, that the distance of $h'$ from $\dxh$ is at most $\delta_{\dec}$, continues to hold under this modified rounding. If this holds, $h'$ can be used to find $h+\ez^{\perp}$ using the unique decoder of $(\dx,\dz)$.
		\begin{align*}
			\Ex{\theta \in [0,1]}{\Delta(h',\dxh)} ~&=~ \Ex{\theta \in [0,1]}{\Ex{\li}{\indi{h'_\li \neq \dxh_\li}}} \\
			~&=~\Ex{\li}{1 - w_{\li \dxh_\li}} \\
			~&=~ \Ex{\li}{\tildeEx{\indi{\unconx(\zee_{\li}) \neq \dxh_{\li}}}} \\
			~&=~ \Ex{\li}{\tildeEx{\indi{\unconx(\zee_{\li}) \neq \unconx(h_{\li})}}} \\
			~&=~ \Ex{\li}{\tildeEx{\indi{\zee_{\li} \not\in h_{\li} + \cz^{\perp} }}} \\
			~&=~ \distLperp{\tildeEx{\cdot}} ~\leq~ \delta_{\dec}\;.
		\end{align*}
		
		The fact that a common threshold $\theta$ shared by all $\li \in L$ is used to generate this rounding allows us to make a single random choice for $\theta$ instead of different random choices for each $\li \in L$. Moreover, the rounded $h'$ is exactly same if the threshold $\theta$ is perturbed without crossing any of the $|L| \cdot t = q^{b_{\out}}\cdot n$ values $\{ w_{\li 1}+w_{\li 2} + \cdots w_{\li j} \}_{\li\in L, j \in [t]}$. Therefore, we only need to try out roundings for $q^{b_{\out}}\cdot n$ many distinct thresholds, which proves the lemma.
	\end{proof}
\end{lemma}

\subsection{Near-MDS Quantum LDPC Codes list decodable up to Johnson bound}\label{sec:near_mds}

In this section, we show how instantiating the AEL amplification with unique decodable asymptotically good QLDPC codes leads to QLDPC codes near the (quantum) Singleton bound that can be list decoded up to the Johnson bound. 

\begin{theorem}[Near-MDS Codes decodable upto Johnson bound]\label{thm:near_mds_main}
	For any $0<\rho<1$, and for any $\nfrac{1}{2} >\eps_1, \eps_2>0$, there is an infinite family of quantum LDPC codes $(\fx,\fz)$ with the following properties:
	\begin{enumerate}[(i)]
		\item The rate of the code is at least $\rho$ and distance is at least $\frac{1-\rho-\eps_1}{2}$.
		\item The code is over an alphabet of size $2^{\calO(\eps_1^{-6}\log(1/\eps_1))}$.
		\item The code of blocklength $n$ can be list decoded from radius $\calJ(\frac{1-\rho - \eps_1}{2})-\eps_2$ in time $n^{\calO_{\eps_1}(1/\eps_2^4)}$.
	\end{enumerate}
\end{theorem}

\begin{proof}
We show how to instantiate \cref{thm:list_decoding_ael} to obtain such codes. 
We will use the AEL distance amplification based on a family of $(n,d,\lambda)$-expanders, with $d$ and $\lambda$ to be chosen later.

The inner code $\cC = (\cx,\cz)$ is chosen to be a quantum Reed-Solomon code of rate $\rho_0 \defeq \frac{\rho}{1-\eps_1}$ and distance $\frac{1-\rho_0}{2}$. Since the inner code has blocklength $d$, the alphabet size of $\cC$ must be at least $d$, and we choose $\cC$ to be defined over a finite field $\F_{2^{b_{\inn}}}$, where $b_{\inn} \leq 1+\log_2 d$. Using \cref{lem:field_downgrade}, we will view $\cC$ as a vector space CSS code with alphabet $\F_2^{b_{\inn}}$.

For the outer code $\cD = (\dx,\dz)$, we fold the binary QLDPC codes of \cite{LZ23:decodable} into blocks of size $b_{\out}$, with $b_{\out} = b_{\inn}\cdot \rho_0 d$ so that concatenation is well defined. The starting binary code from \cite{LZ23:decodable} is chosen to have rate $1-\eps_1$ that can be unique decoded from radius $\delta_{\dec}=\Omega(\eps_1^2)$, and these properties are preserved after folding.

Let $\lambda = \eps_1 \delta_{\dec},\;$ so that $\lambda \leq \Omega(\eps_1^3)$ and $d=\calO(1/\eps_1^6)$. The rate of the AEL-amplified code $(\fx,\fz)$ is $(1-\eps_1) \rho_0 = \rho$, and the distance is at least
\begin{align*}
	\frac{1-\rho_0}{2} - \frac{\lambda}{\delta_1} &~\geq~ \frac{1}{2} \parens*{1-\frac{\rho}{1-\eps_1}} - \frac{\lambda}{\delta_{\dec}} \\
	&~\geq~ \frac{1}{2} \parens*{1-\rho -2\eps_1 \rho} -\eps_1 \\
	&~\geq~ \frac{1-\rho}{2} - 2\eps_1
\end{align*}

The final code is defined over alphabet $\F_2^{b_{\inn}\cdot d}$, which is of size at most $ 2^{(1+\log d))d}  = 2^{ \calO\inparen{\eps_1^{-6}\log(1/\eps_1)}}$. Since AEL amplification preserves the LDPC property of the outer code, our final code $(\fx,\fz)$ is also LDPC. For list decodability, we use \cref{thm:list_decoding_ael} to claim that the above code can be list decoded from $\calJ(\frac{1-\rho}{2}-2\eps_1)-\eps_2$ in time $n^{\calO_{q,d,\delta_{\dec}}(1/\eps_2^4)} = n^{\calO_{\eps_1}(1/\eps_2^4)}$. The claimed parameters can be obtained by replacing $\eps_1$ by $\eps_1/4$. We note that the decoder above can be made deterministic as described in \cref{sec:derandomization}.
\end{proof}

\section*{Acknowledgements}
We are grateful to Louis Golowich for
helpful conversations regarding the quantum version of the AEL procedure, and to Mehdi Soleimanifar
for discussions on LP decoding of quantum codes.
Part of this work was completed while the authors were participating in the ``Analysis and TCS''
program at the Simons Institute in Berkeley. 
We thank the program organizers and Simons administration and staff for their kind hospitality
during this visit.

\bibliographystyle{alphaurl}
\bibliography{macros,madhur2}

\appendix

\section{Correlation Rounding}\label{app:correlation_rounding}

In this appendix, we include proofs of some claims that are standard when using correlation rounding for Sum-of-Squares hierarchy. This presentation is largely borrowed from \cite{JST23}.

Recall that we are dealing with pseudocodewords as defined in \cref{sec:sos_prelims}, which can be seen as relaxations of certain distributions over $\Sigma^E$. Here, $\Sigma$ is an alphabet of size $s$, and $E$ is the edge set of a $d$-regular bipartite graph with $|L|=|R|=n$. 

The following lemma says that the decrease in variance of one random variable upon conditioning on the other is proportional to their covariance. This is based on \cite[Lemma 5.2]{BRS11}, and is the basis of by now a standard technique used for rounding convex hierarchies.
\begin{lemma}
	Let $\tildeEx{\cdot}$ be a pseudoexpectation operator of SoS-degree $t$ with associated pseudovariance $\tildeVar{\cdot}$ and pseudocovariance $\tildeCov{\cdot}{\cdot}$. Assume $S,T$ are sets such that $|S|+|T|\leq t/2$, then,
	\[
		\tildeVar{\zee_S | \zee_T} \leq \tildeVar{\zee_S} - \frac{1}{s^{|T|}} \sum_{\alpha \in \Sigma^S, \beta\in \Sigma^T} \frac{(\tildeCov{\zee_{S,\alpha}}{\zee_{T,\beta}})^2}{\tildeVar{\zee_{T,\beta}}}
	\]
\end{lemma}
Note that since we are only talking about local sets $S$ and $T$ here, the argument follows by working only with the local distributions, which are true distributions.
%

The next lemma gives an averaged out version of the above.
\begin{lemma}\label{lem:conditioning_reduces_variance}
Let $\tildeEx{\cdot}$ be a pseudocodeword of degree at least $4d$, with associated pseudovariance $\tildeVar{\cdot}$ and pseudocovariance $\tildeCov{\cdot}{\cdot}$. Then conditioning on a random $r\in R$ causes a decrease in average pseudovariance in $L$ that is proportional to the average pseudocovariance between $\zee_{\li}$ and $\zee_{r}$ for $\li, \ri$ chosen uniformly from $L\times R$.
    \[
        \Ex{\ri \in R}{\Ex{\li}{ \tildeVar{\zee_{\li} \vert \zee_{\ri}}}} < \Ex{\li}{ \tildeVar{\zee_{\li}}} - \frac{1}{s^{3d}}{\parens*{\Ex{\li,\ri}{\tildeCov{\zee_{\li}}{\zee_{\ri}}}}^2}
    \]
\end{lemma}
\begin{proof}

    \begin{align*}
        \Ex{\ri \in R}{\Ex{\li}{\tildeVar{\zee_{\li}|\zee_{\ri}}}} &= \Ex{\li,\ri}{\tildeVar{\zee_{\li}|\zee_{\ri}}}\\
        &\leq \Ex{\li,\ri}{\tildeVar{\zee_{\li}} - \frac{1}{s^d} \sum_{\alpha, \beta} \frac{(\tildeCov{\zee_{N(\li),\alpha}}{\zee_{N(\ri),\beta}})^2}{\tildeVar{\zee_{N(\ri),\beta}}}} \\
        &\leq \Ex{\li,\ri}{\tildeVar{\zee_{\li}} - \frac{1}{s^d} \sum_{\alpha, \beta} \left( \tildeCov{\zee_{N(\li),\alpha}}{\zee_{N(\ri),\beta}} \right)^2} \\
        &\leq \Ex{\li,\ri}{\tildeVar{\zee_{\li}} - \frac{1}{s^{3d}} \left( \sum_{\alpha, \beta} \abs{ \tildeCov{\zee_{N(\li),\alpha}}{\zee_{N(\ri),\beta}}}\right)^2} \\
        &= \Ex{\li}{\tildeVar{\zee_{\li}}} - \frac{1}{s^{3d}} \Ex{\li,\ri}{\left( \tildeCov{\zee_{\li}}{\zee_{\ri}} \right)^2} \\
        &\leq \Ex{\li}{\tildeVar{\zee_{\li}}} - \frac{1}{s^{3d}} \left( \Ex{\li,\ri}{ \tildeCov{\zee_{\li}}{\zee_{\ri}} } \right)^2 
\qedhere    \end{align*}
\end{proof}
That is, a pseudocodeword which is not $\eta$-good will see its average pseudovariance decrease by $\Omega(\eta^2)$ upon random conditioning. Because the (pseudo-)variance is bounded between 0 and 1, this can be used to obtain an $\eta$-good pseudocodeword after a small constant number of conditionings (although the number of conditionings will depend on $\eta$).
\begin{corollary}\label{lem:low_covariance_solution}
	Given any pseudocodeword $\tildeEx{\cdot}$ of degree $\geq 2d\left(\frac{s^{3d}}{\eta^2}+1\right)$ with associated pseudocovariance operator $\tildeCov{\cdot}{\cdot}$ and an $\eta>0$, there exists a number $u^* \leq s^{3d}/\eta^2 $ such that 
	\[
		\Ex{r_1,r_2,\cdots,r_{u^*}}{\Ex{\li,\ri}{\tildecov[\zee_{\li},\zee_{\ri} \vert \zee_{r_1},\zee_{r_2},\cdots,\zee_{r_{u^*}}]}} \leq \eta
	\]
\end{corollary}

\begin{proof}
	Suppose not, so that 
	\begin{align*}
		\Ex{r_1,r_2,\cdots,r_{u}}{\Ex{\li,\ri}{\tildecov[\zee_{\li},\zee_{\ri} \vert \zee_{r_1},\zee_{r_2},\cdots,\zee_{r_{u^*}}]}} >\eta  && \forall u\in [s^{3d}/\eta^2]
	\end{align*}
	Therefore, by \cref{lem:conditioning_reduces_variance}, for every $u\in [s^{3d}/\eta^2]$,
	\[
		\Ex{r_1,r_2,\cdots,r_{u}}{\Ex{\li}{\tildeVar{\zee_{\li}\big\vert  \zee_{r_1},\zee_{r_2},\cdots,\zee_{r_{u}}}}} < \Ex{r_1,r_2,\cdots,r_{u-1}}{\Ex{\li}{\tildeVar{\zee_{\li}\big\vert  \zee_{r_1},\zee_{r_2},\cdots,\zee_{r_{u-1}}}}} - \frac{1}{s^{3d}}\eta^2 
	\]
	This is impossible, since all pseudovariances are bounded between 0 and 1.
\end{proof}

Finally, the next lemma shows how to use the $\eta$-good property obtained above via random conditioning. Essentially, it allows us to write a pseudoexpectation of products as a product of pseudoexpectations. However, this can only be done when these pseudoexpectations are applied to products coming from a uniform measure on $L\times R$, since that is the kind of guarantee on pseudocovariance that $\eta$-good property gives us.

\begin{lemma}\label{lem:using_eta_good}
	Let $\{X_{\li}\}_{\li\in L}$ and $\{Y_{\ri}\}_{\ri\in R}$ be two collections of $d$-local functions on $\Sigma^E$ such that for every $\li\in L$, $X_{\li}(f)$ only depends on $f_{\li}$ and for every $\ri\in R$, $Y_{\ri}(f)$ only depends on $f_{\ri}$. Further, assume that $\norm{X_{\li}}_{\infty} \leq 1$ and $\norm{Y_{\ri}}_{\infty}\leq 1$ for all $\li \in L, \ri \in R$. Then, for an $\eta$-good pseudocodeword $\tildeEx{\cdot}$,
	\[
		\Ex{\li,\ri}{\tildeEx{X_{\li}(\zee) \cdot Y_{\ri}(\zee)}} ~\leq~ \Ex{\li,\ri}{\tildeEx{X_{\li}(\zee)} \cdot \tildeEx{Y_{\ri}(\zee)}} + \eta 
	\]
\end{lemma}

\begin{proof}
	\begin{align*}
		&\Ex{\li,\ri}{\tildeEx{X_{\li}(\zee) \cdot Y_{\ri}(\zee)} - \tildeEx{X_{\li}(\zee)} \cdot \tildeEx{Y_{\ri}(\zee)}} \\
		&~=~ \Ex{\li,\ri}{\sum_{\alpha \in \Sigma^{N(\li)} \atop \beta\in \Sigma^{N(\ri)}}
                     \tildeEx{X_{\li}(\alpha) \zee_{N(\li),\alpha} \cdot Y_{\ri}(\beta)  \zee_{N(\ri),\beta}} - 
		 \sum_{\alpha \in \Sigma^{N(\li)} \atop \beta\in\Sigma^{N(\ri)}} \cdot \tildeEx{X_{\li}(\alpha) \zee_{N(\li),\alpha}} \tildeEx{Y_{\ri}(\beta) \zee_{N(\ri),\beta}}} \\
		&~=~ \Ex{\li,\ri}{\sum_{\alpha,\beta} X_{\li}(\alpha) Y_{\ri}(\beta) \cdot \inparen{~\tildeEx{
           \zee_{N(\li),\alpha} \cdot \zee_{N(\ri),\beta}} -
           \tildeEx{\zee_{N(\li),\alpha}} \cdot \tildeEx{\zee_{N(\ri),\beta}}~}} \\
		&~\leq~  \Ex{\li,\ri}{\norm{X_{\li}}_{\infty} \norm{Y_{\ri}}_{\infty} \cdot {\sum}_{\alpha,\beta}
           \abs*{~\tildeEx{
           \zee_{N(\li),\alpha} \cdot \zee_{N(\ri),\beta}} -
           \tildeEx{\zee_{N(\li),\alpha}} \cdot \tildeEx{\zee_{N(\ri),\beta}}~}~} \\
		&~\leq~  \Ex{\li,\ri}{\tildecov[\zee_{\li},\zee_{\ri}]} \leq \eta
	\end{align*}		
%
\end{proof}
\section{Expander Mixing Lemma for SoS}\label{app:eml_sos}

In this section, we include proof of the well-known fact that SDPs such as the SoS hierarchy can reason with the expander mixing lemma (EML). We start by giving a version of EML for vector valued functions, which can be obtained by simply applying the usual EML coordinate-wise, along with an application of Cauchy-Schwartz inequality.
\begin{lemma}[EML for vector-valued functions]
	\label{high_dimensional_eml}
		Let $\{v_{\li}\}_{\li \in L}$ and $\{u_{\ri}\}_{\ri \in R}$ be a collection of vectors in $\R^N$. Then,
		\[
			\abs*{\Ex{\li \sim \ri}{\ip{v_{\li}}{u_{\ri}}} - \Ex{\li,\ri}{\ip{v_{\li}}{u_{\ri}}}} \leq \lambda \sqrt{\Ex{\li}{\norm{v_{\li}}^2}} \sqrt{\Ex{\ri}{\norm{u_{\ri}}^2}}
		\]
	\end{lemma}
	
	\begin{proof}
		\begin{align*}
			\abs*{\Ex{\li \sim \ri}{\ip{v_{\li}}{u_{\ri}}} - \Ex{\li,\ri}{\ip{v_{\li}}{u_{\ri}}}} &= \abs*{\Ex{\li \sim \ri}{\Ex{i\in [N]}{v_{\li}(i)u_{\ri}(i)}} - \Ex{\li,\ri}{\Ex{i\in [N]}{v_{\li}(i)u_{\ri}(i)}}} \\
			&\leq \Ex{i}{ \abs*{\Ex{\li \sim \ri}{v_{\li}(i)u_{\ri}(i)} - \Ex{\li,\ri}{v_{\li}(i)u_{\ri}(i)}} } \\
			&\leq \Ex{i}{ \lambda \sqrt{\Ex{\li}{v_{\li}(i)^2}} \cdot \sqrt{\Ex{\ri}{u_{\ri}(i)^2}} } \\
			&\leq \lambda \sqrt{ \Ex{i}{\Ex{\li}{v_{\li}(i)^2}} \cdot \Ex{i}{\Ex{\ri}{u_{\ri}(i)^2}} } \\
			&= \lambda \sqrt{\Ex{\li}{\norm{v_{\li}}}^2 \cdot \Ex{\ri}{\norm{u_{\ri}}^2}}
		\end{align*}
	\end{proof}

Next, we use the non-negativity of sum of squares of low-degree polynomials to show a certain matrix is positive semidefinite, allowing us to use the EML for vector-valued functions to prove the EML for pseudoexpectations. This is the only part in the entire proof that uses the non-negativity of sum-of-squares of low-degree polynomials. The rest of the argument can also be carried out with linear programming hierarchies for example. It is an interesting question whether list decoding can be carried out for AEL without the use of SDPs.
\begin{lemma}[EML for pseudoexpectations]\label{lem:eml_pexp_appendix}
	Let $\{X_{\li}\}_{\li\in L}$ and $\{Y_{\ri}\}_{\ri\in R}$ be two collections of $d$-local functions on $\Sigma^E$ such that for every $\li\in L$, $X_{\li}(f)$ only depends on $f_{\li}$ and for every $\ri\in R$, $Y_{\ri}(f)$ only depends on $f_{\ri}$.
Then for a $\lambda$-spectral expander, we have
	\[
		\abs*{\Ex{\li\sim \ri}{\tildeEx{X_{\li}(\zee) Y_{\ri}(\zee) }} -
                  \Ex{\li,\ri}{\tildeEx{X_{\li}(\zee) Y_{\ri}(\zee)}}} \leq \lambda
                \sqrt{\Ex{\li}{\tildeEx{X_{\li}(\zee)^2}}} \sqrt{\Ex{\ri}
                  {\tildeEx{Y_{\ri}(\zee)^2}}} \mper
	\]
\end{lemma}
\begin{proof}	
Let us use $i$ and $j$ as indices in $[2n]$ that can be used to index $L$ and $R$ via some canonical mapping between $L$ or $R$ and $[n]$. Consider the $2n\times 2n$ matrix $M$, with
\[
		M_{ij} = \begin{cases}
			\tildeEx{X_i(\zee)X_j(\zee)}, & 1\leq i\leq n, 1\leq j\leq n \\
			\tildeEx{X_i(\zee)Y_{j-n}(\zee)}, & 1\leq i\leq n, n+1\leq j\leq 2n \\
			\tildeEx{Y_{i-n}(\zee)X_j(\zee)}, &n+1\leq i\leq 2n, 1\leq j\leq n\\
			\tildeEx{Y_{i-n}(\zee)Y_{j-n}(\zee)} &n+1\leq i\leq 2n,n+1\leq j\leq 2n
		\end{cases}
\]
We will prove that the matrix $M$ is positive semidefinite for any collection of functions as in the lemma statement. To do this, we show that for any vector $v = (x_1,x_2,\cdots,x_n,y_1,y_2,\cdots,y_n)$, the quadratic form $v^TMv$, is non-negative.
	\begin{align*}
		v^TMv &~=~ \Ex{i,j}{M_{ij} x_ix_j + M_{i,j+n}x_iy_j + M_{i+n,j} y_ix_j + M_{i+n,j+n} y_iy_j}\\
		&~=~ \Ex{i,j}{\tildeEx{X_i(\zee)X_j(\zee)}x_ix_j+\tildeEx{X_i(\zee)Y_j(\zee)}x_iy_j+\tildeEx{Y_i(\zee)X_j(\zee)}y_ix_j+\tildeEx{Y_i(\zee)Y_j(\zee)}y_iy_j} \\
		&~=~ \Ex{i,j}{\tildeEx{(x_i X_i(\zee) + y_i Y_i(\zee))(x_jX_j(\zee)+y_jY_j(\zee))}} \\
		&~=~ \tildeEx{\Ex{i}{(x_i X_i(\zee) + y_i Y_i(\zee))}^2}  ~\geq~ 0
	\end{align*}
By Cholesky decomposition, there exist vectors $\{v_{\li}\}_{\li \in L}$ and $\{u_{\ri}\}_{\ri\in R}$ such that 			
\[		
\tildeEx{X_{\li}(\zee) Y_{\ri}(\zee)} = \ip{v_{\li}}{u_{\ri}}, \quad
\tildeEx{X_{\li}(\zee)^2} = \ip{v_{\li}}{v_{\li}}, \quad \text{and} \quad
\tildeEx{Y_{\ri}(\zee)^2} = \ip{u_{\ri}}{u_{\ri}}
\]
We can now apply \cref{high_dimensional_eml} to these vectors.
\[
		\abs*{\Ex{\li\sim \ri}{\tildeEx{X_{\li}(\zee)Y_{\ri}(\zee)}} - \Ex{\li,\ri}{\tildeEx{X_{\li}(\zee)
                      Y_{\ri}(\zee)}}} ~\leq~ \lambda \sqrt{\Ex{\li}{\tildeEx{X_{\li}(\zee)^2}}}
                \sqrt{\Ex{\ri}{\tildeEx{Y_{\ri}(\zee)^2}}} \qquad \qedhere
\]
\end{proof}

\section{Covering Lemma}\label{app:covering}

\begin{lemma}[Covering Lemma]
	Let $g\in \Sigma^E$ and $\alpha \in (0,1)$.
	There exists a distribution $\calD$ over $\ex$ with support size at most $s^d \cdot n + 1$ such that for any $h \in \ex$ such that $\distR{g} < 1-\alpha$, the distribution $\calD$ satisfies
	\[
		\Ex{f\sim \calD}{\distR{f}} < 1-\alpha^2 \mper
	\]
\end{lemma}

\begin{proof}
	Let $\calH = \inbraces{h\in \ex \suchthat \distR{g} <1-\alpha} = \inbraces{h\in \ex \suchthat \ip{\embed(g)}{\embed(h)} > \alpha}$.
	
	Define the set $T = \inbraces{v\in \conv{\embed(\ex)} \suchthat \ip{v}{\embed(g)} >\alpha}$.

\begin{claim}\label{claim:covering}
	Let $g_0 = \argmin_{v\in T} \norm{v}^2$. Then for any $h\in \calH$, it holds that $\ip{g_0}{\embed(h)}>\alpha^2$.
\end{claim}

To finish the proof using this claim, we write $g_0$ as a convex combination of elements in $\embed(\ex)$ and let the corresponding distribution over $\ex$ be $\calD$. By the \Caratheodory theorem, $\calD$ can be chosen with support size at most $s^d \cdot n + 1$. Finally,
\begin{align*}
	\Ex{f\sim \calD}{\distR{f}} &= \Ex{f\sim \calD}{1 - \ip{\embed(f)}{\embed(h)}} \\
	&= 1- \ip{\Ex{f\sim \calD}{\embed(f)}}{\embed(h)} \\
	&= 1- \ip{g_0}{\embed(h)} < 1-\alpha^2 \hfill \qedhere
\end{align*}
\end{proof}

We next prove the claim.
\begin{proof}[Proof of \cref{claim:covering}]
	Note that $\norm{g_0}\cdot \norm{\embed(g)}\geq \ip{g_0}{\embed(g)}>\alpha$ means $\norm{g_0}>\alpha$.

If not, there exists an $h \in \calH$ such that $\ip{g_0}{\embed(h)} \leq \alpha^2$. For a $\tee \in[0,1]$ to be chosen later, consider $g_1 = \tee \cdot g_0+(1-\tee)\cdot \embed(h)$, which is also in $\conv{\embed(\fx)}$, and $\ip{g_1}{\embed(g)} = \tee \cdot \ip{g_0}{\embed(g)}+(1-\tee)\cdot \ip{\embed(h)}{\embed(g)} >\alpha$.
\begin{align*}
	\norm{g_1}^2 &~=~ \ip{\tee \cdot g_0+(1-\tee)\cdot \embed(h)}{\tee \cdot g_0+(1-\tee)\cdot \embed(h)} \\
	&~=~ \tee^2 \cdot \norm{g_0}^2+2\tee(1-\tee) \cdot \ip{g_0}{\embed(h)}+(1-\tee)^2 \cdot \norm{\embed(h)}^2\\
	&~\leq~ \tee^2 \cdot \norm{g_0}^2+2\tee(1-\tee)\cdot \alpha^2+(1-\tee)^2
\end{align*}
The optimum of this quadratic function is strictly less than $\norm{g_0}^2$ if $\norm{g_0}>\alpha$, which is indeed the case as $\norm{g_0}\cdot \norm{\embed(g)}\geq \ip{g_0}{\embed(g)}>\alpha$ means $\norm{g_0}>\alpha$.
\end{proof}
%

\end{document}